\documentclass[11pt,reqno]{article}
\usepackage[margin=1in]{geometry}
\usepackage{authblk}

\usepackage[english]{babel}
\usepackage[dvipsnames]{xcolor}
\usepackage{amsmath, amsthm, amssymb}
\usepackage{array, booktabs}
\usepackage{graphicx}
\usepackage{mathtools}
\usepackage{microtype}
\usepackage{caption}
\usepackage{float}
\usepackage{todonotes}
\usepackage[title]{appendix}

\usepackage[
  colorlinks=true,
  linkcolor=ForestGreen,
  citecolor=NavyBlue,
  filecolor=Plum,
  urlcolor=Plum,
  backref=page
]{hyperref}
\usepackage{fourier}
\usepackage{algorithm, algpseudocode, algorithmicx}

\algrenewcommand\alglinenumber[1]{\sf\scriptsize\color{NavyBlue}{#1}}
\usepackage{backref}
\renewcommand*{\backref}[1]{}
\renewcommand*{\backrefalt}[4]{%
	\ifcase #1 %
	(No citations.)
	\or
	(Cited p.~#2.)
	\else
	(Cited pp.~#2.)
	\fi
}
\DeclarePairedDelimiter{\fronorm}{\lVert}{\rVert_{\rm F}}
\usepackage[capitalise,nameinlink]{cleveref}
\crefname{equation}{}{}
\numberwithin{equation}{section}
\newcommand{\phimax}{\phi_{\max}}
\newcommand{\Bbad}{\widetilde{\bB}}
\newcommand{\Id}{\bI}

\newcommand{\ones}{\mathbf{1}}
\newcommand{\zeros}{\mathbf{0}}

\DeclareMathOperator{\E}{\mathbb{E}}

\newcommand{\R}{\mathbb{R}}

\newcommand{\bv}[1]{\mathbf{#1}}

\newcommand{\bw}{\mathbf{w}}
\newcommand{\bx}{\mathbf{x}}

\newcommand{\bA}{\mathbf{A}}
\newcommand{\bB}{\mathbf{B}}

\newcommand{\bI}{\mathbf{I}}

\newcommand{\bN}{\mathbf{N}}
\newcommand{\bP}{\mathbf{P}}

\newcommand{\bR}{\mathbf{R}}
\newcommand{\bS}{\mathbf{S}}

\newcommand{\bX}{\mathbf{X}}
\newcommand{\bY}{\mathbf{Y}}

\usepackage{thmtools}
\declaretheorem[name=Theorem]{theorem}
\declaretheorem[name=Fact,numberlike=theorem]{fact}

\declaretheorem[name=Lemma,numberlike=theorem]{lemma}

\declaretheorem[name=Proposition,numberlike=theorem]{proposition}

\theoremstyle{definition}
\declaretheorem[name=Definition,numberlike=theorem]{definition}
\newtheorem*{definition*}{Definition}

\theoremstyle{remark}

\DeclarePairedDelimiter{\norm}{\lVert}{\rVert}

\DeclareMathOperator*{\tr}{tr}

\DeclareMathOperator*{\Binomial}{Binomial}
\DeclareMathOperator*{\Poisson}{Poisson}

\makeatletter
\newcommand\fs@paperruled{%
  \def\@fs@pre{\hrule height 0.4pt \kern 4pt}%
  \def\@fs@mid{\kern 4pt \hrule height 0.4pt \kern 4pt}%
  \def\@fs@post{\kern 4pt \hrule height 0.4pt \relax}%
  \let\@fs@capt\floatc@ruled
  \def\@fs@cfont{\bfseries}%
  \let\@fs@iftopcapt\iftrue
}
\makeatother

\newcolumntype{L}[1]{>{\raggedright\arraybackslash}p{#1}}
\newcolumntype{C}[1]{>{\centering\arraybackslash}p{#1}}
  \usepackage{nth}
  \usepackage{intcalc}

\title{\vspace{-1.5em}Fast Length-Squared Sampling for Positive-Semidefinite Matrices}
\author[1]{Rajarshi Bhattacharjee}
\author[2]{Ethan N. Epperly}
\author[1]{Cameron Musco}
\author[1]{Aaron Tian}
\affil[1]{University of Massachusetts Amherst}
\affil[2]{University of California Berkeley}

\date{}

\begin{document}
\maketitle
\begin{abstract}
We describe a simple rejection-sampling-based algorithm to perform \textit{length-squared sampling} on an $n \times n$ positive-semidefinite (psd) matrix: that is, to sample a column with probability proportional to its squared \(\ell_2\)-norm. The algorithm runs in just $O(n)$ expected time, which is significantly sublinear in the input matrix size.
The runtime is optimal, even when the input is assumed to be diagonal. 

Our result has several applications.
Length-squared sampling is used by a number of sublinear time algorithms for matrix problems, like low-rank approximation and eigenvalue approximation. Often, it is assumed that the algorithm is given access to the matrix column norms, and thus can perform length-squared sampling efficiently. Our result shows that, at least for psd matrices, we can remove this assumption. We also discuss an application to an asymptotically optimal algorithm for estimating the Frobenius norm of a psd matrix to relative error. 
Finally, we show that our sampling algorithm yields a very simple sublinear time algorithm for the \emph{robust psd low-rank approximation} problem introduced in \cite{bakshi_robust_2020}, which nearly matches the more complex method developed there. 
\end{abstract}

\section{Introduction}
\label{sec:intro}

In modern applications, we frequently encounter matrices of massive scale that would be intractable to store and compute with directly. A few examples include Hessians in neural network optimization \cite{martens2010deep}, interaction kernels in particle simulations \cite{greengard1987fast}, massive Kronecker-structured operators in quantum physics and other scientific applications \cite{oseledets2011tensor,schollwock2011density}, and \emph{kernel matrices} in machine learning and spatial statistics \cite{SS02,rasmussen2003gaussian}. 
In this last example, given data points \(\bx_1,\ldots,\bx_n \in \R^d\) and a positive definite kernel function \(k : \R^d\times \R^d \to \R\), one must perform computations involving the kernel matrix
$
\bA \in \R^{n\times n}$ with entries $\bA_{ij} = k(\bx_i,\bx_j).$
In modern settings, when $n$ can be in the billions \cite{MCRR20}, even forming and storing all $n^2$ entries of $\bA \in \R^{n\times n}$, let alone performing linear algebraic manipulations on it, is infeasible. 

Fortunately, it is often relatively cheap to access individual matrix entries \(\bA_{ij}\) -- e.g., in the kernel setting, this typically requires just one $O(d)$ time kernel function computation. This opens the possibility of processing large matrices using \emph{sublinear time algorithms} that access only a small fraction of the input matrix entries (i.e., use $o(n^2)$ total accesses).

There is a large body of research on sublinear matrix algorithms -- see, e.g., \cite{williams_using_2000,frieze_fast_2004,drineas_nystrom_2005,deshpande_matrix_2006,drineas_fast_2006,drineas_fast_2006-1,SV09,CD13,musco_recursive_2017,musco_sublinear_2019,bakshi_robust_2020,chen_randomly_2025,CY25,GW26}.
A frequently used primitive is \emph{length-squared sampling}, in which a row or column of the input matrix is sampled with probability proportional to its squared \(\ell_2\) norm.
This primitive is used, e.g., in classic algorithms for approximate matrix multiplication  \cite{drineas_fast_2006} and low-rank approximation \cite{frieze_fast_2004,drineas_fast_2006-1}. It is also used by the randomized Kaczmarz method for solving linear systems \cite{SV09,needell_stochastic_2015} and in sublinear time algorithms for eigenvalue approximation \cite{bhattacharjee_sublinear_2022,swartworth_tight_2024}, distance matrix approximation \cite{bakshi2018sublinear}, and more.
Unfortunately, length-squared sampling is not always available in the sublinear setting: computing even one column norm requires reading \(n\) entries, and computing the full sampling distribution requires reading the entire matrix.
Thus, to employ length-squared sampling, we either need to have the row or column norms provided to us, or we must exploit additional structure in the matrix.

\subsection{Our Results}

Our main contribution is to show that there is an extremely simple and efficient sublinear time algorithm for length-squared sampling on any positive semidefinite (psd) matrix $\bA$. Formally, we prove:

\begin{theorem}[Sublinear time length-squared sampling]
\label{thm:intro-fast-l2-sampling}
Let \(\bA\in\R^{n\times n}\) be nonzero and psd.
There exists an algorithm (\cref{alg:sampling-procedure}) that outputs a random index \(I\in[n]\) sampled from the length-squared distribution
\begin{equation} \label{eq:length-squared}
    \Pr[I = i]
    =
    \frac{\|\bA_{*,i}\|_2^2}{\fronorm{\bA}^2}\qquad\text{for all $i\in[n]$}.
\end{equation}
The algorithm reads \(O(n)\) entries of $\bv A$ and performs $O(n)$ arithmetic operations in expectation.
\end{theorem}

\cref{alg:sampling-procedure} uses a simple rejection sampling scheme where two indices $I$ and $J$ are sampled from the diagonal sampling distribution $\mathcal D(i) = \bA_{ii} / \tr(\bA)$ and a proposal $(I,J)$ is accepted with probability $\bA_{IJ}^2 / (\bA_{II} \bA_{JJ})$. This is a valid probability when $\bv A$ is psd -- see \Cref{fact:psd-inequality}.
Conditional on acceptance, $(I,J)$ is drawn from squared-entry distribution $\mathcal{E}(i,j) = \bA_{ij}^2 / \fronorm{\bA}^2$ and the second index $J$ thus follows the length-squared distribution \cref{eq:length-squared}. In fact, by symmetry, both indices follow this distribution. The key proof step is to show that the acceptance probability is $\Omega(1/n)$, and so the algorithm runs in $O(n)$ expected time. 

For a general matrix, performing length-squared sampling requires $\Omega(n^2)$ entry accesses: consider $\bA$ with a single nonzero entry in a random position. Thus, to achieve $O(n)$ runtime, our algorithm must exploit psd structure. Beyond being a mathematically natural class, psd matrices appear frequently in applications.
Indeed, the kernel matrix associated with any positive definite kernel function $k$ and any data set $\{ \bx_i \}$ is always psd.

\smallskip

\noindent \textbf{Optimality.}
It is not hard to see that \Cref{thm:intro-fast-l2-sampling} is asymptotically optimal: the class of psd matrices includes all positive diagonal matrices. To perform length-squared sampling from such a matrix, we must at least read the full diagonal, requiring $\Omega(n)$ time. In many applications, one requires multiple indices sampled independently from the length-squared distribution. Simply repeating our algorithm $m$ times samples $m$ such indices in time $O(mn)$. In \Cref{thm:optimality}, we show that this is also optimal -- i.e., one cannot sample many indices in $o(n)$ time per index on average.

\smallskip

\noindent\textbf{Frobenius Norm Estimation.}
Using similar ideas to our sublinear-time length-squared sampling algorithm, we also present an algorithm for computing a $(1\pm \epsilon)$-relative error approximation to the Frobenius norm of a psd matrix in $O(n/\epsilon^2)$ time. Formally, we prove:

\begin{theorem}[Sublinear time Frobenius norm approximation]
\label{thm:frobenius-norm-intro}
Let \(\bA\in\R^{n\times n}\) be a psd matrix. For any \(\epsilon,\delta\in(0,1)\) there exists an algorithm (\Cref{alg:frobenius-norm-approximation}) that outputs an estimate \(\widehat{F}\) satisfying
\[
    (1-\epsilon)\fronorm{\bA}^2
    \leq
    \widehat{F}
    \leq
    (1+\epsilon)\fronorm{\bA}^2 \quad \text{with probability } \ge 1-\delta.
\]
The algorithm reads \(O(n\log(1/\delta)/\epsilon^2)\) entries of $\bv A$ and performs \(O(n\log(1/\delta)/\epsilon^2)\) arithmetic operations.
\end{theorem}
We show in \Cref{thm:frob_norm_optimality} that the complexity of \Cref{thm:frobenius-norm-intro} is also asymptotically optimal: any algorithm estimating the Frobenius norm of psd $\bv A$ to $(1\pm \epsilon)$-relative error with constant probability must read $\Omega(n/\epsilon^2)$ entries in expectation.

\smallskip

\noindent\textbf{$\ell_p$ Norm Sampling.} We additionally show in \Cref{thm:fast-lp-sampling} that a small modification of our algorithm can sample columns according to the $\ell_p$ column-norm distribution:
\[
    \Pr[I=i] = \frac{\norm{\bv A_{*,i}}_p^p}{\sum_{k=1}^n \norm{\bv A_{*,k}}_p^p}.
\]
The algorithm runs in $O(n)$ expected time for any $p \geq 0$. While not as fundamental as length-squared sampling, $\ell_p$ norm sampling is also commonly applied in randomized algorithms for linear algebra \cite{dasgupta2009sampling,munteanuturnstile}.

\subsection{Applications}

Using our fast length-squared sampling and Frobenius norm approximation algorithms, we derive several sublinear time algorithms for linear algebra problems involving psd matrices.

\begin{itemize}

    \item \textbf{Eigenvalue approximation.} We show that, for psd matrices, the sublinear time eigenvalue approximation algorithm of \cite{bhattacharjee_sublinear_2022,swartworth_tight_2024} can be implemented using entry queries alone, without assuming access to the matrix column norms as input.
    The resulting algorithm approximates all eigenvalues of a psd matrix $\bA$ to additive error \(\pm \epsilon\fronorm{\bA}\) using \(\widetilde O(n/\epsilon^2)\) entry queries in expectation.

    \item \textbf{Additive-error low-rank approximation.} 
    We describe a simple algorithm that outputs the factors  $\bv X,\bv Y \in \R^{n \times k}$ of a rank-$k$ approximation $\bX\bY^\top$  to a psd matrix $\bA$ satisfying the additive-error guarantee
    \[
        \fronorm{\bA - \bX\bY^\top}^2 \le \fronorm{\bA-\bA_k}^2+\epsilon\fronorm{\bA}^2,
    \]
    where $\bv A_k$ is the optimal rank-$k$ approximation, given by the singular value decomposition.
    Our algorithm uses \(O(nk/\epsilon^2)\) entry queries and \(O(nk^2/\epsilon^4)\) arithmetic operations in expectation.
    It is based on an algorithm of \cite{drineas_fast_2006-1} that assumes access to the column norms. It is much simpler than known sublinear algorithms for psd low-rank approximation with Frobenius norm error bounds \cite{musco_sublinear_2019,bakshi_robust_2020}.
    However, these prior algorithms either achieve somewhat improved sample complexity and/or stronger relative-error guarantees; see \cref{sec:additive-lra} for discussion.

    \item \textbf{Robust low-rank approximation.} 
    Finally, we extend the above result to  a robust version of the psd low-rank approximation problem introduced by \cite{bakshi_robust_2020}. Here, we are given query access to a corrupted matrix $\bB = \bA + \bN$ where $\bv A$ is psd and $\bv N$ is arbitrary noise with  $\fronorm{\bN}^2 \le \eta \fronorm{\bA}^2$.
    We provide a simple algorithm for producing the factors of a rank-$k$ approximation $\bX\bY^\top$ satisfying
    \begin{equation*}
        \fronorm{\bA - \bX\bY^\top}^2 \le \fronorm{\bA - \bA_k}^2 + (\epsilon + O(\sqrt{\eta})) \fronorm{\bA}^2.
    \end{equation*}
    The algorithm makes $O(\phimax^2 n k /\epsilon^2)$ queries in expectation, where $\phimax \ge \max\{ 1, \max_{i \in [n]} \bA_{ii} / |\bB_{ii}|\}$ is a given corruption parameter. Our error bound matches what is achieved by the algorithm of \cite{bakshi_robust_2020}. Further, our algorithm is significantly simpler, and removes an assumption required by \cite{bakshi_robust_2020} on the row norms of $\bv N$. Our runtime also matches \cite{bakshi_robust_2020} up to a suboptimal dependence on the accuracy parameter $\epsilon$.
\end{itemize}
Across these three applications, our goal is to demonstrate the applicability of our length-squared sampling algorithm in giving (relatively) simple sublinear algorithms for different linear algebraic tasks that make minimal additional assumptions on the input matrix.

\subsection{Related Work}

Length-squared sampling is one of the most widely used primitives in randomized matrix computations, dating back to Frieze, Kannan, and Vempala's seminal work on fast low-rank approximation \cite{frieze_fast_2004}.
Length-squared sampling has also been used in significant  follow-up work on low-rank approximation \cite{drineas_fast_2006-1,DKM06b,deshpande_matrix_2006,deshpande_adaptive_2006}, for approximate matrix multiplication \cite{drineas_fast_2006}, for the iterative solution of linear systems via the Kaczmarz method \cite{SV09}, and in work on sublinear time eigenvalue and eigenvector approximation \cite{bhattacharjee_sublinear_2022,swartworth_tight_2024,RB26}.
Length- and entry-squared sampling has also been a fundamental tool in dequantization and quantum-inspired algorithm design \cite{Tan19,CGL+22,chepurko2022quantum}.

Motivated by kernel-based machine learning and Gaussian process-based methods in statistics, there has been significant interest in developing sublinear time algorithms for psd matrices in particular.
This research spans theoretical computer science \cite{drineas_nystrom_2005,musco_recursive_2017,musco_sublinear_2019,bakshi_robust_2020,DM20,bhattacharjee_sublinear_2022,swartworth_tight_2024}, applied mathematics \cite{FL24,chen_randomly_2025,ETW25}, and machine learning  \cite{FS01,rahimi2007random,le2013fastfood,Bac13a,RCR15,RCCR19}. Existing approaches employ techniques such as uniform sampling \cite{williams_using_2000}, diagonal and diagonal-power sampling \cite{drineas_nystrom_2005,chen_randomly_2025}, ridge leverage score approximation \cite{musco_recursive_2017,RCCR19,musco_sublinear_2019,bakshi_robust_2020}, determinantal sampling \cite{DM20}, and random featured-based methods \cite{rahimi2007random}.
Most closely related to our approach, several papers on sublinear time low-rank approximation of psd and distance matrices use random sampling approaches that either explicitly or implicitly estimate the squared column norm distribution \cite{bakshi2018sublinear,indyk2019sample,bakshi_robust_2020}. However, no prior work yields a sublinear time algorithm that actually samples from this distribution, as we do in \Cref{thm:intro-fast-l2-sampling}.

\subsection{Notation and Access Model}

Throughout, we let $[n]$ denote the set of indices $\{1,\ldots,n\}$.
For a matrix \(\bA\), we let \(\bA_{ij}\) denote the \((i,j)\)-entry, \(\bA_{i,*}\) denote the \(i\)th row, and \(\bA_{*,i}\) denote the \(i\)th column.
We write \(\fronorm{\bA}\) for the Frobenius norm, \(\tr(\bA) = \sum_{i \in [n]} \bv A_{ii} \) for the trace when $\bv A$ is square, and $\bA_k$ for the Eckart--Young best rank-$k$ approximation, given by projection onto $\bv A$'s top-$k$ singular vectors. $\bv A_k$ minimizes $\norm{\bv A - \bv A_k}_*$ over all rank-$k$ matrices, where $\norm{\cdot}_*$ is any unitarily invariant norm, including the Frobenius norm.
The $\ell_2$ norm of a vector $\bx \in \R^n$ is denoted $\norm{\bx}$.
A symmetric matrix \(\bA\in\R^{n\times n}\) is positive semidefinite (psd) if \(\bx^\top\bA\bx\geq 0\) for all \(\bx\in\R^n\).

Throughout, we use the entry query model where we interact with a matrix $\bA$ by querying individual entries. We report algorithm costs in terms of both the number of entry queries and the total number of arithmetic operations.

\section{Fast Length-Squared Sampling}
\label{sec:fast-sample}
We start by presenting our main algorithm (\Cref{alg:sampling-procedure}) for length-squared sampling from psd matrices, and bounding its runtime. As discussed, the algorithm in fact gives a bit more than length-squared sampling: it samples a pair of indices $(I,J)$ from the entry-squared distribution $\mathcal E(i,j) = \bA_{ij}^2 / \fronorm{\bA}^2$.
Focusing on either $I$ or $J$ in isolation, we obtain a sample from the length-squared distribution. 

To sample from $\mathcal E(i,j)$, the algorithm performs rejection sampling: it first samples $I$ and $J$ independently from the diagonal distribution $\mathcal D(i) = \bA_{ii} / \tr(\bA)$ (line 3 of \Cref{alg:sampling-procedure}).
That is, we have $(I,J) = (i,j)$  with probability $\bv A_{ii} \bv A_{jj}/\tr(\bv A)^2$. A proposal $(I,J)$ is then accepted with probability $\bA_{IJ}^2 / (\bA_{II} \bA_{JJ})$ (line 4 of \Cref{alg:sampling-procedure}). This ensures that the pair is sampled with probability proportional to $\bv A_{IJ}^2$, guaranteeing correctness of the algorithm. The key proof step is to argue that the acceptance probability is $\Omega(1/n)$, and thus the number of proposals required in expectation (and the runtime) is $O(n)$.  

\begin{algorithm}[H]
\caption{Fast entry-squared and length-squared sampling} \label{alg:sampling-procedure}
\begin{algorithmic}[1]
\Require Nonzero psd matrix $\bA \in \R^{n\times n}$.
\Ensure Index $(I,J)$ sampled from the entry-squared distribution $\mathcal E(i,j) = \bA_{ij}^2 / \fronorm{\bA}^2$. The indices $I$ and $J$ are sampled individually from the length-squared distribution $\mathcal{L}(i) = \norm{\bA_{i,*}}^2 / \fronorm{\bA}^2$.
\State Read the diagonal of \(\bA\), and define the distribution
\(
    \mathcal D(i)={\bA_{ii}}/{\tr(\bA)}
\)
for all $i\in[n].$
\While{\texttt{true}}
\State Sample independent $I,J \sim \mathcal D$.
\State {With probability} $\bA_{IJ}^2 / (\bA_{II}^{\vphantom{2}} \bA_{JJ}^{\vphantom{2}})$ \textbf{return} $I, J$.
\EndWhile
\end{algorithmic}
\end{algorithm}
\vspace{-.5em}
For the acceptance probability in line 4 of \Cref{alg:sampling-procedure} to even be a valid probability in $[0,1]$, we need the following \emph{off-diagonal inequality} for psd matrices, which states that the $(i,j)$-entry is bounded in magnitude by the geometric mean of the $i$th and $j$th diagonal entries.

\begin{fact}[Off-diagonal inequality]
    \label{fact:psd-inequality}
    Let \(\bA\in\R^{n\times n}\) be psd. For all \(i,j\in[n]\),
    \[
        \bA_{ij}^2\leq \bA_{ii}^{\vphantom{2}}\bA_{jj}^{\vphantom{2}}.
    \]
\end{fact}

\Cref{fact:psd-inequality} is easy to prove: since $\bA$ is psd, the $\{i,j\}$-principal submatrix is psd as well, so its determinant $\bA_{ii}^{\vphantom{2}} \bA_{jj}^{\vphantom{2}} - \bA_{ij}^2$ is nonnegative. Intuitively, the off-diagonal inequality is critical to achieving sublinear time length-squared sampling for psd matrices.
In particular, it shows that the largest entry always lies on the diagonal, and large off-diagonal entries can only appear in rows and columns with large diagonals. This rules out the $\Omega(n^2)$ hard case for length-squared sampling from general matrices, where a single nonzero entry is hidden in an arbitrary location and must be found.

We are now ready to prove our main guarantee for sampling from the entry-square distribution with \cref{alg:sampling-procedure}. The result implies \Cref{thm:intro-fast-l2-sampling} as an immediate corollary. 

\begin{theorem}[Sublinear time entry-squared sampling]
    \label{thm:fast-l2-sampling} Let $\bv A \in \R^{n \times n}$ be a nonzero psd matrix.
    Let $(I, J)$ denote the pair returned by \cref{alg:sampling-procedure} on input $\bv A$. For all $i,j\in [n]$, we have
    \[
        \Pr[I = i, J = j]
        =
        \frac{\bA_{ij}^2}{\fronorm{\bA}^2}.
    \]
    In turn,
    $
        \Pr[I = i]
        =
        \frac{\|\bA_{*,i}\|_2^2}{\fronorm{\bA}^2} 
    $ and $
        \Pr[J = j]
        =
        \frac{\|\bA_{*,j}\|_2^2}{\fronorm{\bA}^2}.
    $
    The algorithm requires \(2n\) entry queries and $O(n)$ arithmetic operations in expectation.
    Moreover, with probability $1-\delta$, the algorithm terminates using at most $n + \lceil n \log(1/\delta) \rceil$ entry accesses and $O(n \log(1/\delta))$ arithmetic operations.
\end{theorem}

\begin{proof}
    Since $\bv A$ is psd and nonzero, it has nonnegative diagonal entries and $\tr(\bv A) > 0$. Thus, the diagonal distribution $\mathcal D(i)={\bA_{ii}}/{\tr(\bA)}$ computed in line 1 is well-defined. 
    The off-diagonal inequality (\cref{fact:psd-inequality}) ensures that
    \(
        {\bA_{ij}^2}/({\bA_{ii}\bA_{jj}})\in[0,1]
    \)
    for all $i,j \in [n]$, so the acceptance probability in line 4 is also well-defined.

    Fix a pair \((i,j)\).
    In one round, the probability that the algorithm returns \((i,j)\) is
    \[
        \frac{\bA_{ii}}{\tr(\bA)}
        \cdot
        \frac{\bA_{jj}}{\tr(\bA)}
        \cdot
        \frac{\bA_{ij}^2}{\bA_{ii}\bA_{jj}}
        =
        \frac{\bA_{ij}^2}{\tr(\bA)^2}.
    \]
    Summing over all pairs $(i,j)$, the total probability of returning in each round is
    \[
        \beta
        =
        \sum_{i,j\in[n]}
        \frac{\bA_{ij}^2}{\tr(\bA)^2}
        =
        \frac{\fronorm{\bA}^2}{\tr(\bA)^2}.
    \]
    Thus, conditioning on acceptance, the probability of outputting \((i,j)\) is
    \[
        \frac{\bA_{ij}^2/\tr(\bA)^2}{\beta}
        =
        \frac{\bA_{ij}^2}{\fronorm{\bA}^2},
    \]
    as claimed.
    Summing over the first coordinate $i$ gives
    \[
        \Pr[J = j]
        =
        \sum_{i\in[n]}
        \frac{\bA_{ij}^2}{\fronorm{\bA}^2}
        =
        \frac{\|\bA_{*,j}\|_2^2}{\fronorm{\bA}^2}.
    \]
    Similarly, summing over the second coordinate gives $\Pr[I = i] = \frac{\|\bA_{i,*}\|_2^2}{\fronorm{\bA}^2} = \frac{\|\bA_{*,i}\|_2^2}{\fronorm{\bA}^2}$, where the second equality follows by symmetry of $\bv A$.
    
    Finally, we obtain a lower bound on the success probability for one round. We have:
    \begin{align*}        
        \beta
        =
        \frac{\fronorm{\bA}^2}{\tr(\bA)^2} \ge
        \frac{\sum_{i=1}^n \bv A_{ii}^2}
        {\left(\sum_{i=1}^n \bv A_{ii} \right)^2}\geq
        \frac{1}{n}\qquad\text{by Cauchy--Schwarz.}
    \end{align*}

    The number of rejection sampling rounds needed is a geometric random variable with mean $1/\beta$.
    Consequently, the expected number of rounds is
    \(1/\beta\leq n\) and the algorithm terminates in $R$ rounds except with probability $(1-\beta)^{R} \le \mathrm{e}^{-\beta R} \le \mathrm{e}^{-R/n}$.
    The algorithm reads the diagonal once in line 1 to form
    \(\mathcal{D}\), requiring $n$ entry queries. In each round, it uses one additional entry query and a constant number of arithmetic operations.
    The stated resource estimates follow.
\end{proof}

Note that if one desires an algorithm with a fixed runtime upper bound, one can always terminate \Cref{alg:sampling-procedure} after $O(n \log(1/\delta))$ iterations. If the algorithm has not returned a sample before then, which will happen with probability at most $\delta$, an arbitrary index can be output. This yields an algorithm that samples from a distribution that is $\delta$-close in total variation distance to the true entry-squared (or length-squared) distribution. In most applications, exact length-squared sampling is not required -- it typically suffices to use any sampler that outputs $i \in [n]$ with probability $\ge \frac{c\|\bA_{*,i}\|_2^2}{\fronorm{\bA}^2}$ for some constant $c > 0$ \cite{drineas_fast_2006-1}. Setting $\delta$ to a constant suffices for such a guarantee to hold.

Also note that  our rejection sampling approach can also be used to sample from the $p$th-power entry distribution $\mathcal{E}_p(i,j) = |\bA_{ij}|^p / \sum_{k,\ell} |\bA_{k\ell}|^p$ and perform $\ell_p$ column-norm sampling tasks with $O(n)$ complexity; see \Cref{app:lp-sampling} -- the algorithm and proof are essentially identical to \Cref{alg:sampling-procedure} and \Cref{thm:fast-l2-sampling}.

\subsection{Optimality}

The $O(n)$ runtime of \Cref{thm:fast-l2-sampling} is clearly optimal -- even if $\bv A$ is just a positive diagonal matrix, we must read its full diagonal if we hope to perform length-squared sampling. For example, there could be a single nonzero diagonal entry, which we must identify. Indeed, \Cref{alg:sampling-procedure} starts by reading the full diagonal in line 1. 

Many applications, however, require taking multiple length-squared samples from the same matrix. Since the cost of reading the diagonal can be amortized across the samples, one might hope that $o(n)$ amortized time is possible in this setting. 
Unfortunately, we prove that this is not the case -- for any $m < n$, taking $m$ independent samples from the length-squared distribution requires $\Omega(mn)$ time. That is, one can do no better than simply repeating \Cref{alg:sampling-procedure} $m$ times. Formally:

\begin{theorem}[Length-squared sampling: Optimality] \label{thm:optimality}
    Any algorithm that, given any nonzero psd matrix $\bv A \in \R^{n \times n}$ outputs $m$  indices
    $i_1,\ldots,i_m\in[n]$ sampled independently from the length-squared distribution \cref{eq:length-squared}
    must read 
    \(
        \Omega\!\left(\min\{mn,n^2\}\right)
    \)
    entries of $\bv A$ in expectation in the worst case.
\end{theorem}

\begin{proof}
    We assume $n \ge 1000$ and $10 \le m \le n/100$, proving a lower bound of $\Omega(mn)$ in this setting, which yields the theorem. 
    Fix $s \coloneqq \lceil 2\sqrt{n/m}\rceil$.
    Our assumptions ensure $s\le n$.
    Consider the random hard instance 
    \begin{equation*}
        \bA = \bP\begin{bmatrix}
            \ones_{s\times s} & \zeros_{s\times (n-s)} \\
            \zeros_{(n-s)\times s} & \Id_{n-s}
        \end{bmatrix}\bP^\top,
    \end{equation*}
    where $\ones_{(\cdot)}$ and $\zeros_{(\cdot)}$ denote matrices of all ones and zeros, respectively, of the appropriate dimensions, $\bv I_{n-s}$ is the $n-s \times n-s$ identity matrix, and $\bP$ is a uniformly random permutation matrix.
    Call a column of $\bA$ \emph{heavy} if it has squared norm $s$. There are $s$ heavy columns. All other columns have unit norm.
    
    Intuitively, since $\bv A$ has just $O(s^2)$ nonzero off-diagonal entries in random locations, identifying one of these entries, and thus identifying one of the heavy columns with good probability, requires reading $\Omega(n^2/s^2) = \Omega(nm)$ entries. At the same time, the sum of the squared lengths of the heavy columns is $s^2 = \Theta(n/m) = \Theta(\fronorm{\bv A}^2/m)$, so $m$ length-squared samples suffice to identify a heavy column with good probability. Thus, forming $m$ length-squared samples must require reading $\Omega(nm)$ entries of $\bv A$.
    
    Formally, consider the following \emph{heavy column problem}: output a subset $S \subseteq [n]$ of $m$ elements such that at least one element is heavy.
    We reduce the heavy column problem to length-squared sampling.
    Indeed, given access to an algorithm that produces $m$ samples from the length-squared distribution using $t$ queries in expectation, one can solve the heavy column problem with probability at least $2/3$ using $6t$ queries in expectation.
    By Markov's inequality, the algorithm succeeds in $6t$ queries with probability at least $5/6$.
    Conditional on the algorithm succeeding, each output sample is heavy with probability
    \[
        \frac{s^2}{\fronorm{\bv A}^2}
        = \frac{s^2}{s^2+n-s}
        \geq \frac{4n/m}{4n/m+n}
        = \frac{4}{m+4}
        \geq \frac{2.5}{m},
    \]
    where in the second inequality we use our assumption that $m \ge 10$ and thus $m+4 \le 8m/5$.

    Thus, none of the returned columns is heavy with probability at most $(1 - 2.5/m)^m \le \mathrm{e}^{-2.5} < 1/6$.
    Overall, the algorithm thus solves the heavy column problem with probability at least $2/3$.

    We next show that any algorithm solving the heavy column problem with probability at least $2/3$ must use $\Omega(mn)$ entry accesses in expectation.
    By Yao's minimax principle, we are free to assume without loss of generality that the algorithm is deterministic.
    Define the following sets:
    \begin{enumerate}
        \item Let $N_0\subseteq [n]^2$ be the set of coordinates queried by the deterministic algorithm on input $\bI_n$.
        \item Let $S_0\subseteq[n]$ be the set of elements output by the deterministic algorithm on input $\bI_n$.
    \end{enumerate}
    Now, suppose we run the deterministic algorithm on $\bA$ drawn from our hard input distribution, and let $\mathcal{E}$ be the event that the algorithm queries an off-diagonal 1. On $\overline{\mathcal{E}}$, the algorithm receives
    exactly the same answers as on $\Id_n$ and hence outputs $S_0$. 
    We have the chain of inequalities
    \begin{multline*}
        \frac{2}{3} \le \Pr[\text{algorithm succeeds}] = \Pr[\text{algorithm succeeds} \mid \mathcal{E}] \cdot \Pr[\mathcal{E}] +  \Pr[\text{algorithm succeeds} \mid \overline{\mathcal{E}}] \cdot \Pr[\overline{\mathcal{E}}] \\
        \le \Pr[\mathcal{E}] + \Pr[\text{algorithm succeeds} \mid \overline{\mathcal{E}}] = \Pr[\mathcal{E}] + \Pr[\text{$S_0$ contains a heavy element}].
    \end{multline*}
    We can bound the probabilities of these events as follows:
    \begin{align*}
        &\Pr[\mathcal{E}] \le \sum_{(i,j) \in N_0} \Pr[\text{$i$ is heavy and $j$ is heavy}] \le |N_0|\cdot \frac{s(s-1)}{n(n-1)}; \\
        &\Pr[\text{$S_0$ contains a heavy element}] \le \sum_{i \in S_0} \Pr[\text{$i$ is heavy}] = |S_0| \cdot \frac{s}{n} = \frac{ms}{n}.
    \end{align*}
    Rearranging, we thus have
    \begin{align*}
    |N_0|\cdot \frac{s(s-1)}{n(n-1)} + \frac{ms}{n} &\ge \frac{2}{3}\\
        |N_0| &\ge \frac{n(n-1)}{s(s-1)} \cdot \left(\frac{2}{3} - \frac{ms}{n}\right) \\
        &=\Theta(mn).
    \end{align*}
    Therefore, any algorithm solving the heavy column problem with probability at least $2/3$ on a worst-case input must make $\Omega(mn)$ queries in expectation, completing the proof.
\end{proof}

\section{Frobenius Norm Approximation}
\label{sec:frobenius-norm-approximation}

The ideas of \cref{alg:sampling-procedure} naturally extend to yield a simple sublinear time algorithm (\Cref{alg:frobenius-norm-approximation}) for approximating the Frobenius norm of psd $\bv A$ to relative error.
The basic idea is that, for $I,J$ sampled  independently from the diagonal sampling distribution $\mathcal{D}(i) = \bA_{ii} / \tr(\bA)$, the quantity $X \coloneqq \bA_{IJ}^2 / \bA_{II} \bA_{JJ}$ forms an unbiased estimate for the ratio $\fronorm{\bA}^2 / \tr(\bA)^2$.
By averaging iid copies of $X$, we obtain an estimator that concentrates around its expected value. Rescaling by $\tr(\bv A)^2$ gives an estimator for $\fronorm{\bA}^2$.

\begin{algorithm}[H]
\caption{Frobenius Norm Approximation} \label{alg:frobenius-norm-approximation}
\begin{algorithmic}[1]
\Require Nonzero psd matrix $\bA \in \R^{n\times n}$, number of trials $s$.
\Ensure Estimate $\widehat{F}$ for $\fronorm{\bA}^2$.
\State Read the diagonal of $\bA$ and define $\mathcal{D}(i) = \bA_{ii} / \tr(\bA)$ for all $i \in [n]$.
\For{$k = 1, \ldots, s$}
\State Sample independent $I,J\sim \mathcal{D}$.
\State $X_k \gets \bA_{IJ}^2 / (\bA_{II}\bA_{JJ})$.
\EndFor
\State \Return $\widehat{F} \coloneqq \tr(\bA)^2 \cdot \frac{1}{s} \sum_{k = 1}^s X_k$
\end{algorithmic}
\end{algorithm}

\begin{theorem}[Frobenius norm approximation]
\label{thm:frobenius-norm-approximation-algorithm}
Let \(\bA\in\R^{n\times n}\) be a nonzero psd matrix. For any \(\epsilon,\delta\in(0,1)\), \cref{alg:frobenius-norm-approximation} with $s = \lceil 3n\log(2/\delta)/\epsilon^2\rceil$ produces an estimate \(\widehat{F}\) satisfying
\[
    (1-\epsilon)\fronorm{\bA}^2
    \leq
    \widehat{F}
    \leq
    (1+\epsilon)\fronorm{\bA}^2 \quad \text{with probability } \ge 1-\delta.
\]
The algorithm requires $O(n+s)=O(n\log(1/\delta)/\epsilon^2)$ entry queries and arithmetic operations. 
\end{theorem}

\begin{proof}
The random variables $X_1,\ldots,X_s$ are iid, with mean 
\begin{equation*}
    \E[X_k] = \sum_{i,j \in [n]} \frac{\bA_{ij}^2}{\bA_{ii}\bA_{jj}} \cdot \frac{\bA_{ii}}{\tr(\bA)}\cdot \frac{\bA_{jj}}{\tr(\bA)} = \frac{\fronorm{\bA}^2}{\tr(\bA)^2}.
\end{equation*}
Moreover, the off-diagonal inequality (\Cref{fact:psd-inequality}) ensures that $X_k \in [0,1]$.
Thus, by a standard Chernoff bound \cite[App.~A]{CR14a},
\begin{equation*}
    \Pr\left[\, \left|\widehat{F} - \fronorm{\bA}^2\right| \ge \epsilon \fronorm{\bA}^2 \,\right] \le 2 \exp\left(-\frac{s\fronorm{\bA}^2\epsilon^2}{3\tr(\bA)^2}\right) \quad \text{for any } \epsilon \in (0,1).
\end{equation*}
Setting $s = \lceil 3n\log(2/\delta)/\epsilon^2\rceil$ and using the inequality $\tr(\bA) \le \sqrt{n}\fronorm{\bA}$, we conclude that $\left|\widehat{F} - \fronorm{\bA}^2\right| \le \epsilon \fronorm{\bA}^2$ with probability at least $1-\delta$, as desired.
\end{proof}

\subsection{Optimality}

We next show that \Cref{alg:frobenius-norm-approximation} is asymptotically optimal: any algorithm estimating the Frobenius norm of a psd input matrix $\bv A \in \R^{n \times n}$ to $(1\pm \epsilon)$ relative error with good probability must read $\Omega(n/\epsilon^2)$ entries of $\bv A$ in expectation. Formally:

\begin{theorem}[Frobenius norm approximation: Optimality]\label{thm:frob_norm_optimality}
    Any algorithm that, given a nonzero psd matrix $\bA\in\R^{n\times n}$ outputs an estimate $\widehat F$ satisfying for $\epsilon \in (0,1/2)$
    \[
    (1-\epsilon)\fronorm{\bA}^2
    \leq
    \widehat{F}
    \leq
    (1+\epsilon)\fronorm{\bA}^2 \quad \text{with probability } \ge 2/3
    \]
    must read $\Omega( \min \{n/\epsilon^2,n^2\})$ entries of $\bA$ in the worst case.
\end{theorem}

\begin{proof}[Proof of \cref{thm:frob_norm_optimality}]
The intuition behind the lower bound is similar to that of \Cref{thm:optimality}.  We consider a psd input matrix with $\Theta(n)$ randomly placed nonzero off-diagonal entries with random values.
Finding any one of these off-diagonal entries requires reading $\Omega(n)$ entries of $\bv A$ in expectation. Further, by standard sampling lower bounds \cite{canetti1995lower}, we must read $\Omega(1/\epsilon^2)$ of the nonzero entries to estimate their average, and in turn $\fronorm{\bA}^2$, to $(1\pm \epsilon)$-relative error. Overall, we must thus read $\Omega(n/\epsilon^2)$ entries.

Formally, we consider random matrices of the form
\[
    \bB(\bw)
    \coloneqq
    \begin{bmatrix}
        \Id_{n/2} & \operatorname{diag}(\bw)\\
        \operatorname{diag}(\bw) & \Id_{n/2}
    \end{bmatrix},
    \qquad
    \bA({\bw})
    \coloneqq
    \bP\bB(\bw)\bP^\top,
\]
for a vector $\bw\in\R^{n/2}$ and uniformly random permutation matrix $\bP$. 
Let $\bw_0\in\R^{n/2}$ have $(1/4-3\epsilon)n$ entries set to 1 uniformly at random and the rest set to 1/2, and let $\bw_1\in\R^{n/2}$ have $(1/4 + 3\epsilon)n$ entries set to 1 uniformly at random and the rest set to 1/2. Finally, let $\bA_0 = \bA({\bw_0})$ and $\bA_1 = \bA({\bw_1})$, and assume $\epsilon \ge C/\sqrt{n}$ for a sufficiently large positive constant $C > 0$.

Consider the following problem: given entry query access to a random matrix $\bA$, which is equal to either $\bA_0$ or $\bA_1$, each with probability 1/2, decide whether the given matrix is $\bA_0$ or $\bA_1$ with success probability at least 2/3. We check that (assuming $\epsilon \le 1/2$), $(1-\epsilon)\fronorm{\bA_1}^2 > (1+\epsilon)\fronorm{\bA_0}^2$, so a $(1\pm\epsilon)$ approximation to the Frobenius norm of the given matrix is enough to distinguish between $\bA_0$ and $\bA_1$.

First, note that $\bv w_0$ and $\bv w_1$ are random $\{1/2,1\}$ vectors whose means differ by $O(\epsilon)$ -- standard bounds on mean estimation thus imply that we must read $\Omega(1/\epsilon^2)$ entries of $\bv w$ to distinguish these vectors, and in turn, to distinguish $\bv A_0$ from $\bv A_1$. See e.g., Theorem 1 of \cite{canetti1995lower}.\footnote{Theorem 1 of \cite{canetti1995lower} gives a lower bound for a slightly more general problem -- where the vectors may have entries in $[0,1]$. But one can check that the hard case used in their proof is indeed equivalent to ours -- $\bv w_0$ and $\bv w_1$ are Boolean with means $1/2\pm \epsilon$.}
Thus, to complete the argument, all we must do is show that it takes $\Omega(n/\epsilon^2)$ queries to identify $\Omega(1/\epsilon^2)$ distinct nonzero off-diagonal entries. 

By Yao's minimax principle, we may assume without loss of generality that the algorithm is deterministic. We also assume without loss of generality that the algorithm never queries an unordered pair $\{i,j\}$ more than once, as doing so reveals no new information. Given integers $r$ and $q$, let $\mathcal{E}$ be the event that the algorithm, run on $\bA$, receives at least $r$ nonzero off-diagonal values within its first $q$ queries. To every occurrence of $\mathcal{E}$, we attach a certificate of the form $t_1 < t_2 < \cdots < t_r, v_1, v_2, \ldots, v_r$, where $t_i$ is the query time and $v_i$ is the value of the $i$th nonzero response. There are $\binom{q}{r}$ ways to form the $t_1, \ldots, t_r$ and $2^r$ choices for $v_1, \ldots, v_r$, since the $v_i$ take on values in $\{1/2, 1\}$. This gives at most $2^r\binom{q}{r}$ distinct certificates.

Now, bound the probability of observing any given certificate. We first note that a certificate $t_1 < t_2 < \cdots < t_r, v_1, v_2, \ldots, v_r$ uniquely determines the first $t_r + 1$ queries made by the algorithm, since it is deterministic. Fixing a certificate, we may therefore recover the unordered pairs $\{i_1, j_1\}, \ldots, \{i_r, j_r\}$ queried by the algorithm at times $t_1, \ldots, t_r$. We note the pairs must be pairwise disjoint, since there is exactly one nonzero off-diagonal entry in every row or column of $\bB(\bw)$, and hence $\bA(\bw)$. We bound the probability that a uniformly random permutation $\bP$ sends a nonzero off-diagonal entry to each of $\{i_1, j_1\}, \ldots, \{i_r, j_r\}$ as
\[
    \frac{1}{(n-1)(n-3)\cdots(n-2r+1)}\leq \left(\frac{2}{n}\right)^r
\]
for $r\leq n/4$. Taking a union bound over all possible certificates gives
\[\Pr[\mathcal{E}]\leq 2^r\binom{q}{r}\left(\frac{2}{n}\right)^r\leq\left(\frac{4\mathrm{e}q}{rn}\right)^r\]
where we used the bound $\binom{q}{r}\leq (\mathrm{e}q/r)^r$. In particular, if $r = \Omega(1/\epsilon^2)$ and $\Pr[\mathcal{E}] = \Omega(1)$, then $q = \Omega(n/\epsilon^2)$, completing the argument.
\end{proof}

\section{Applications}
\label{sec:applications}
We next present several applications of our sublinear time length-squared sampling and Frobenius norm estimation algorithms. In particular, we discuss results on eigenvalue approximation for psd matrices (\Cref{sec:eigenvalue-approximation}), additive error Frobenius norm low-rank approximation for psd matrices (\Cref{sec:additive-lra}), and robust psd low-rank approximation (\Cref{sec:robust-lra}).


\subsection{Eigenvalue Approximation}
\label{sec:eigenvalue-approximation}
We first give a sublinear time algorithm to compute additive-error approximations to all eigenvalues of a psd matrix. Formally, we target the following form of approximation, studied in \cite{bhattacharjee_sublinear_2022,swartworth_tight_2024}: 
\begin{definition}[Additive-error eigenvalue approximation]
\label{def:additive-eigenvalue-approximation}
Let \(\bA\in\R^{n\times n}\) be symmetric with eigenvalues $\lambda_1(\bv A) \geq \ldots \ge \lambda_n(\bv A)$. We say that a sequence
\(
    \widehat{\lambda}_1
    \geq
    \widehat{\lambda}_2
    \geq
    \cdots
    \geq
    \widehat{\lambda}_n
\)
is an additive \(\gamma\)-approximation to the spectrum of \(\bA\) if
\[
    |\widehat{\lambda}_i-\lambda_i(\bA)|
    \leq
    \gamma
    \qquad
    \text{for all } i\in[n].
\]
\end{definition}
%
Our algorithm will compute an additive $(\varepsilon \fronorm{\bA})$-eigenvalue approximation.
This is a strong guarantee in that we must provide an approximation to each eigenvalue of $\bA$. However, one can observe that at most $\varepsilon^{-2}$ eigenvalues can be larger than $\epsilon \fronorm{\bA}$, so an additive-error eigenvalue algorithm is free to return zero as an approximation to most of the spectrum. 

We will use the basic algorithmic template of~\cite{bhattacharjee_sublinear_2022}, which was refined in~\cite{swartworth_tight_2024}. These papers study two related algorithms. The first samples and rescales a $\tilde O(1/\epsilon^2) \times \tilde O(1/\epsilon^2)$ principal submatrix uniformly at random and returns the eigenvalues of the submatrix padded with additional zeros. \cite{bhattacharjee_sublinear_2022,swartworth_tight_2024} show that this method achieves an $\epsilon n \|\bv A \|_{\infty}$-eigenvalue approximation, where $\norm{\bv A}_\infty$ is the maximum magnitude of an entry of $\bv A$. For the second approach, they assume that they can perform length-squared sampling of rows/columns of $\bv A$ and have knowledge of $\|\bv A \|_F$ and all column norms. In this case, after sampling and rescaling a $\tilde O(\log^4 n/\epsilon^2) \times \tilde O(\log^4 n/\epsilon^2)$ principal submatrix using length-squared sampling, and zeroing out some entries to control variance (which requires knowledge of the Frobenius norm and column norms), they again return the eigenvalues of the zeroed-out matrix along with additional zeros. This method is shown to give an  $\epsilon  \|\bv A \|_{F}$-eigenvalue approximation. Note that $\|\bv A \|_{F} \leq  n \|\bv A \|_{\infty} $, and in fact, $\|\bv A \|_{F}$ is usually much smaller than $n \|\bv A \|_{\infty} $. So the length-squared sampling approach can be much stronger than the uniform sampling approach.

In the special case that $\bv A$ is psd, we can apply  Algorithms~\ref{alg:sampling-procedure} and~\ref{alg:frobenius-norm-approximation} to implement the length-squared sampling approach of \cite{swartworth_tight_2024} with total sample complexity $\tilde O(n/\epsilon^2)$. The pseudocode for our implementation is given in~\cref{alg:sublinear-psd-eigenvalue} in~\cref{sec:eigenvalue-proof}. Our implementation and analysis require several small tweaks. First, ~\cite{swartworth_tight_2024} requires knowledge of $\| \bv A\|_F^2$ for their zeroing-out step; we show that a $(1 \pm \epsilon)$-approximation to $\| \bv A\|_F^2$ computed via \Cref{alg:frobenius-norm-approximation} suffices. Moreover, while we can sample indices from the length-squared distribution, we can only rescale them approximately, since we don't know the exact normalizing factor $\norm{\bv A}_F^2$ of the sampling probabilities. Again, we show that a $(1\pm \epsilon)$ approximation suffices. Finally, in Algorithm 2 of~\cite{swartworth_tight_2024} they independently sample each row $\textrm{Binomial} \left(m, \frac{\|\bv A_{i,*} \|_2^2}{\| \bv A\|_F^2} \right)$ many times. It is not clear how to sample from this distribution using repeated independent length-squared samples, which \Cref{alg:sampling-procedure} provides. We thus use a standard Poissonization trick: we generate our number of samples $K$ according to the $\textrm{Poisson}(m)$ distribution, and then sample $K$ times independently from the length-squared distribution. This is equivalent to sampling each row independently $\textrm{Poisson}\left ( \frac{m\|\bv A_{i,*} \|_2^2}{\| \bv A\|_F^2} \right )$ times, which suffices. Such a sampling scheme was first proposed and analyzed in~\cite{RB26}, where they extended the sampling algorithms of \cite{bhattacharjee_sublinear_2022,swartworth_tight_2024} to also approximate eigenvectors. 
We now state the main result of this section, with the proof and details deferred to~\cref{sec:eigenvalue-proof}.

\begin{theorem}[Eigenvalue approximation]
\label{thm:sublinear-psd-eigenvalue}
Given any psd matrix \(\bA\in\R^{n\times n}\) and accuracy parameter $\epsilon \in (0,1)$, \cref{alg:sublinear-psd-eigenvalue} returns an \((\epsilon\|\bA\|_F)\)-additive approximation to the spectrum of \(\bA\) with probability at least \(2/3\). The algorithm queries
\[
    O\!\left(
        \frac{n}{\epsilon^2}
        \log^4 n
        \log^2\frac{1}{\epsilon}
    \right)
\]
entries of \(\bA\) in expectation, and uses
\[
    O\!\left(
        \frac{n}{\epsilon^2}
        \log^4 n
        \log^2\frac{1}{\epsilon}
        +
        \frac{1}{\epsilon^{2\omega}}
        \log^{4 \omega} n
        \log^{2\omega}\frac{1}{\epsilon}
    \right)
\]
arithmetic operations in expectation, where $\omega$ is the exponent of fast matrix multiplication.
\end{theorem}
We note that the success probability of \Cref{thm:sublinear-psd-eigenvalue} can easily be boosted via the median trick. 

\subsection{Additive-Error Low-Rank Approximation}
\label{sec:additive-lra}
In this section, we derive a simple sublinear time algorithm for computing an additive error low-rank approximation to a psd matrix $ \bv A \in \R^{n \times n}$ using our fast length-squared sampling procedure. It is well known that length-squared sampling can be used to compute additive error low-rank approximations \cite{frieze_fast_2004,drineas_fast_2006-1,bakshi_robust_2020}. We apply a result of Drineas, Kannan, and Mahoney~\cite{drineas_fast_2006-1} in particular, which lets us compute an approximate basis for the top \(k\)-dimensional left singular subspace of \(\bA\) from \(O(k/\epsilon^2)\) columns sampled via the length-squared distribution. We then use a well-known approximate matrix multiplication result to approximately project \(\bA\) onto this subspace. Our algorithm outputs \(\bX, \bY\in\R^{n\times k}\) satisfying
\[\fronorm{\bA - \bX\bY^\top}^2\leq \fronorm{\bA - \bA_k}^2 + \epsilon\fronorm{\bA}^2\]
with constant probability, using \(O(nk/\epsilon^2)\) entry queries to \(\bA\) and \(O(n(k/\epsilon^2)^{\omega-1})\) arithmetic operations in expectation, where \(\omega\) is the exponent of fast matrix multiplication. 

Our sample complexity incurs an additional \(1/\epsilon\) factor in comparison to the best known algorithm for low-rank approximation: Bakshi, Chepurko, and Woodruff~\cite[Theorem 4.1]{bakshi_robust_2020} achieve a relative error low-rank approximation guarantee of $(1+\epsilon)\fronorm{\bv A - \bv A_k}^2$ using only \(\widetilde{O}(nk/\epsilon)\) entries, matching a \(\Omega(nk/\epsilon)\) query lower bound of Musco and Woodruff~\cite{musco_sublinear_2019} up to polylogarithmic factors. While their relative error guarantee is a strictly stronger result, the algorithm presented here is extremely simple and admits a simple analysis. We pose as an open question whether it is enough to sample only \(O(k/\epsilon)\) columns by squared column norms, thus leading to an \(O(nk/\epsilon)\) algorithm. The comparison is summarized in \cref{tab:additive-lra-comparison}.

\begin{table}[ht!]
\centering
\small
\renewcommand{\arraystretch}{1.25}
\resizebox{\textwidth}{!}{%
\begin{tabular}{@{}L{0.15\textwidth}L{0.23\textwidth}L{0.35\textwidth}C{0.2\textwidth}C{0.2\textwidth}@{}}
\toprule
Result
&
Assumptions
&
Guarantee
&
Query complexity
\\
\midrule
\cite{bakshi_robust_2020}, Theorem 4.1
&
Psd matrix with entry-wise access.
&
\(\bX,\bY\in\R^{n\times k}\) such that
\[
    \fronorm{\bA-\bX\bY^\top}^2
    \leq
    (1+\epsilon)\fronorm{\bA-\bA_k}^2.
\]
&
\[
    \widetilde{O}\!\left(
        \frac{nk}{\epsilon}
    \right)
\]
\\
\midrule
This work
&
Psd matrix with entry-wise access.
&
\(\bX,\bY\in\R^{n\times k}\) such that
\[
    \fronorm{\bA-\bX\bY^\top}^2
    \leq
    \fronorm{\bA-\bA_k}^2
    +
    \epsilon\fronorm{\bA}^2 .
\]
&
\[
    O\!\left(
        \frac{nk}{\epsilon^2}
    \right)
\]
\\
\bottomrule
\end{tabular}
}
\caption{Comparison of sublinear-time low-rank approximation guarantees for
psd matrices.}
\label{tab:additive-lra-comparison}
\end{table}

We begin with the following theorem due to Drineas, Kannan, and Mahoney~\cite{drineas_fast_2006-1}, which shows that one can compute a good approximation to the top \(k\)-dimensional left singular subspace of \(\bA\) from \(O(k/\epsilon^2)\) columns sampled by squared column norms.
Here is a simplified version of their result.

\begin{lemma}[Additive-error low-rank approximation: General matrix; \cite{drineas_fast_2006-1}, Theorem~4]
\label{thm:dkm06}
    Let \(\bA\in\R^{m\times n}\) be a nonzero matrix, let $k \leq n$ be a rank parameter, and let $0 \le \epsilon \le 1$ and $0<\delta \le 1/\mathrm{e}$ be error and failure-probability parameters.
    Suppose $I_1,\ldots,I_c$ are i.i.d.\ samples from the length-squared sampling distribution \cref{eq:length-squared}, and let $\bX \in \R^{m\times k}$ be the top-$k$ singular vectors of the column-normalized submatrix 
    \begin{equation*}
        \begin{bmatrix}
            \dfrac{\bA_{*,I_1}}{\norm{\bA_{*,I_1}}} & \cdots & \dfrac{\bA_{*,I_c}}{\norm{\bA_{*,I_c}}}
        \end{bmatrix} \in \R^{m\times c}.
    \end{equation*}
    Then if $c \ge 72\log(1/\delta) \cdot k/\epsilon^2$, we have 
    $\fronorm{\bA - \bX\bX^\top\bA}^2\leq \fronorm{\bA - \bA_k}^2 + \epsilon \fronorm{\bA}^2$ with probability $\ge 1-\delta.$
\end{lemma}

Lemma~\ref{thm:dkm06} produces a matrix \(\bX \in \R^{m \times k}\) with orthonormal columns such that \(\bX\bX^\top\bA\), the projection of \(\bA\) onto the subspace spanned by \(\bX\), yields an additive-error rank-$k$ approximation to \(\bA\). To produce an explicit factorization, we would like to use \(\bX\) as the left factor and set \(\bY_{\rm opt}^\top = \bX^\top \bA\) so that \(\bX\bY_{\rm opt}^\top = \bv X \bv X^\top \bv A\). However, computing \(\bX^\top\bA\) requires us to read all of \(\bA\), so we instead apply a standard approximate matrix multiplication result to approximate \(\bX^\top \bA\). We state the needed guarantee below, which is due to Drineas, Kannan, and Mahoney~\cite{drineas_fast_2006}, and provide a proof in \Cref{app:probabilistic-right-factor-proof} for completeness.

\begin{lemma}[Approximate matrix multiplication; \cite{drineas_fast_2006}]
\label{lem:probabilistic-right-factor}
Let \(\bA\in\R^{n\times n}\) be a matrix, let
\(\bX\in\R^{n\times k}\) have orthonormal columns, and let $I_1,\ldots,I_t$ be i.i.d.\ samples from the leverage score (i.e., squared row-norm) distribution of $\bX$:
\begin{equation*}
    q_i \coloneqq \frac{\|\bX_{i,*}\|_2^2}{k} \quad \text{for } i = 1,\ldots, n.
\end{equation*}
Then for every $\delta\in(0,1)$, the matrix
\[
    \bY^\top  
    \coloneqq
    \frac{1}{t}\sum_{r=1}^t
    \frac{\bX_{I_r,*}^{\top}\bA_{I_r,*}}{q_{I_r}}
\]
satisfies the bound 
\[
    \fronorm{\bA-\bX\bY^\top}^2
    \leq
    \fronorm{\bA-\bX\bX^\top\bA}^2
    +
    \frac{k}{\delta t}\fronorm{\bA}^2 \quad\text{with probability }\ge 1 - \delta.
\]
\end{lemma}

Combining length-squared sampling to form the left factor $\bX$ (\cref{thm:dkm06}) with approximate matrix multiplication to form the right factor $\bY^\top$ (\cref{lem:probabilistic-right-factor}) yields \cref{alg:sublinear-psd-low-rank} for additive-error sublinear psd low-rank approximation.
Its performance is quantified by the following theorem.

\begin{algorithm}[t]
\caption{Sublinear Psd Low-Rank Approximation} \label{alg:sublinear-psd-low-rank}
\begin{algorithmic}[1]
\Require Psd matrix $\bA \in \R^{n\times n}$, target rank $k \in [n]$, and accuracy parameter $\epsilon \in (0,1)$.
\Ensure Matrices $\bX \in \R^{n\times k}$ and $\bY \in \R^{n\times k}$ such that $\bX\bY^\top \approx \bA$.
\State Set $s = \lceil Ck/\epsilon^2 \rceil$ and $t = \lceil C'k/\epsilon \rceil$ for sufficiently large absolute constants $C, C' > 0$.
\State Run \cref{alg:sampling-procedure} $s$ times to sample i.i.d.\ indices $I_1,\ldots,I_s$ from the squared-length distribution \label{line:length-squared-sample}
\State Read the full columns $\bA_{*,I_1},\ldots,\bA_{*,I_s}$ and form $\bB \coloneqq \begin{bmatrix}
    \bA_{*,I_1}/\norm{\bA_{*,I_1}} & \cdots & \bA_{*,I_s}/\norm{\bA_{*,I_s}}
\end{bmatrix} \in \R^{n\times s}$
\State Compute the top-$k$ left singular vectors $\bX \in \R^{n\times k}$ of $\bB$
\State Let $\mathcal{S}(i) = \norm{\bX_{i,*}}^2/k$ be the leverage score distribution on $\bX$ and sample $J_1,\ldots,J_t \overset{\mathrm{i.i.d.}}{\sim} \mathcal{S}$.
\State Query the full rows $\bA_{J_1,*},\ldots,\bA_{J_t,*}$ and set
\(
    \bY \coloneqq \frac{1}{t}\sum_{r=1}^{t} \frac{\bA_{J_r,*}^{\top}\bX_{J_r,*}}{\norm{\bX_{J_r,*}}^2/k} \in \R^{n\times k}
\) 
\State \Return $\bX$, $\bY$
\end{algorithmic}
\end{algorithm}

\begin{theorem}[Additive-error low-rank approximation: Psd matrix]
\label{thm:sublinear-psd-low-rank}
Let \(\bA\in\R^{n\times n}\) be a nonzero psd matrix, let
\(1\leq k\leq n\), and let \(\epsilon\in(0,1)\).
\Cref{alg:sublinear-psd-low-rank} returns matrices \(\bX,\bY\in\R^{n\times k}\) satisfying
\[
    \fronorm{\bA-\bX\bY^\top}^2
    \leq
    \fronorm{\bA-\bA_k}^2+\epsilon\fronorm{\bA}^2
\]
with probability at least \(2/3\). The algorithm queries \(O(nk/\epsilon^2)\) entries of \(\bA\) and uses \(O(n(k/\epsilon^2)^{\omega-1})\) arithmetic operations in expectation.
\end{theorem}

\begin{proof}
We shall show the guarantee holds by choosing the constants $C = 517$ and $C' = 12$ in \cref{alg:sublinear-psd-low-rank}. 
First note that by choice of $C$, the hypotheses of \cref{thm:dkm06} are satisfied with failure probability $\delta = 1/6$ and accuracy parameter $\epsilon / 2$.
Thus,
\begin{equation} \label{eq:lra-guarantee-1}
    \fronorm{\bA - \bX\bX^\top\bA}^2 \le \fronorm{\bA - \bA_k}^2 + \frac{\epsilon}{2} \fronorm{\bA}^2 \quad \text{with probability } \ge \frac{5}{6}.
\end{equation}
Next, the choice of $C' = 12$ shows that 
\begin{equation} \label{eq:lra-guarantee-2}
    \fronorm{\bA - \bX\bY^\top}^2 \le \fronorm{\bA - \bX\bX^\top\bA}^2 + \frac{k}{(1/6)(12k/\epsilon)}\fronorm{\bA}^2 = \fronorm{\bA - \bX\bX^\top\bA}^2 + \frac{\epsilon}{2}\fronorm{\bA}^2 \quad \text{with probability } \ge \frac{5}{6}.
\end{equation}
Combining \cref{eq:lra-guarantee-1,eq:lra-guarantee-2} yields the desired guarantee.

The algorithm accesses $s+t$ full rows or columns and executes \cref{alg:sampling-procedure} a total of $s$ times, resulting in an expected $O(n(s+t)) = O(nk/\epsilon^2)$ entry accesses.
The runtime is dominated by computing the top-$k$ singular vectors of $\bB$, which requires $O(ns^{\omega-1}) = O(n(k/\epsilon^2)^{\omega-1})$ operations via a fast-matrix-multiplication-accelerated SVD \cite{Demmel_2007}. 
\end{proof}

\subsection{Robust Low-Rank Approximation}
\label{sec:robust-lra}
In the previous section, we used \cref{alg:sampling-procedure} to obtain an additive-error low-rank approximation to a psd matrix in sublinear time. In this section, we show that the same approach is robust to a natural class of entry-wise perturbations. Our starting point is the robust low-rank approximation model of Bakshi, Chepurko, and Woodruff~\cite{bakshi_robust_2020}. In this model, there is an underlying psd matrix \(\bA\), but we are only given entry-wise access to a corrupted matrix
    $\bB = \bA + \bN$,
where \(\bN\) is an arbitrary noise matrix with bounded Frobenius norm. Our goal is to compute a low-rank approximation to \(\bA\).

Bakshi, Chepurko, and Woodruff \cite{bakshi_robust_2020} consider this problem under three assumptions on the corruption. First, they assume that the noise has bounded Frobenius norm. That is, for some $\eta \le 1$, 
\begin{equation} \label{eq:assumption-1}
    \fronorm{\bN}^2 \leq \eta \fronorm{\bA}^2. \tag{A1}
\end{equation}
Second, they assume that the noise is well-spread across the rows. That is, for some constant $c > 0$,
\begin{equation} \label{eq:assumption-2}
    \norm{\bN_{i,*}}^2
    \leq
    c\norm{\bA_{i,*}}^2
    \qquad\text{for every }i\in[n].\tag{A2}
\end{equation} 
Third, they assume access to the diagonal corruption parameter
\begin{equation} \label{eq:assumption-3}
    \phimax
    \coloneqq
    \max\left\{1,\,\max_{i\in[n]} \frac{\bA_{ii}}{|\bB_{ii}|}\right\}.\tag{A3}
\end{equation}
Under these assumptions, their algorithm returns rank-\(k\) approximation factors \(\bX, \bY\in\R^{n\times k}\) satisfying
\begin{align}\label{eq:robustGuarantee}
    \fronorm{\bA-\bX\bY^\top}^2
    \leq
    \fronorm{\bA-\bA_k}^2
    +
    \bigl(\epsilon + O(\sqrt{\eta})\bigr)\fronorm{\bA}^2
\end{align}
with constant probability using
\(
    \widetilde{O}\left({\phimax^2 n k}/{\epsilon}\right)
\)
entry queries to \(\bB\). 

The main result of this section is a simple algorithm that obtains the same guarantee as \eqref{eq:robustGuarantee} under assumptions \cref{eq:assumption-1,eq:assumption-3}. Notably, it does not need the row-wise noise assumption \cref{eq:assumption-2}. Our approach is to modify our sampling algorithm (\Cref{alg:sampling-procedure}) to obtain length-squared samples from \(\bB\) at the cost of \(O(\phimax^2n)\) expected queries per sample. We then apply the additive error low-rank approximation algorithm (\Cref{alg:sublinear-psd-low-rank}) of \cref{sec:additive-lra} to obtain a low-rank approximation to \(\bB\), and thus to \(\bA\). Since by \Cref{thm:sublinear-psd-low-rank}, \Cref{alg:sublinear-psd-low-rank} requires $O(1/\epsilon^2)$ length-squared samples, our algorithm queries
\(
    O\left({\phimax^2 n k}/{\epsilon^2}\right)
\)
entries of \(\bB\) in expectation, so we incur an additional factor of \(1/\epsilon\) in the query complexity compared to \cite{bakshi_robust_2020}. The comparison is summarized in \cref{tab:robust-lra-comparison}.

\begin{table}[t]
\centering
\begin{tabular}{llll}
\toprule
Result
&
Assumptions
&
Guarantee
&
Query complexity

\\
\midrule
\cite{bakshi_robust_2020}
&
\cref{eq:assumption-1}, \cref{eq:assumption-2}, and \cref{eq:assumption-3}
&
\(
\begin{aligned}
    \fronorm{\bA-\bX\bY^\top}^2
    &\leq
    \fronorm{\bA-\bA_k}^2 \\
    &\quad+
    \bigl(\epsilon+O(\sqrt{\eta})\bigr)\fronorm{\bA}^2
\end{aligned}
\)
&
\(
    \widetilde{O}\!\left(
        \dfrac{\phimax^2 n k}{\epsilon}
    \right)
\)

\\
\midrule
This work
&
\cref{eq:assumption-1,eq:assumption-3}
&
\(
\begin{aligned}
    \fronorm{\bA-\bX\bY^\top}^2
    &\leq
    \fronorm{\bA-\bA_k}^2 \\
    &\quad+
    (\epsilon+O(\sqrt{\eta}))\fronorm{\bA}^2
\end{aligned}
\)
&
\(
    O\!\left(
        \dfrac{\phimax^2 n k}{\epsilon^2}
    \right)
\)

\\
\bottomrule
\end{tabular}
\caption{Comparison of robust additive-error low-rank approximation guarantees.
All error guarantees hold with constant probability.}
\label{tab:robust-lra-comparison}
\end{table}

The main challenge in implementing the above approach is showing how to obtain length-squared samples from \(\bB\) with \(O(\phimax^2n)\) expected queries per sample. Since \(\bB\) may not be psd, \cref{alg:sampling-procedure} cannot be applied directly. However, we can apply a similar algorithm, after \emph{clipping} the entries of \(\bB\), following an idea of \cite{bakshi_robust_2020}, which we recall here.
\begin{definition}[Clipped matrix]
\label{def:clipped-matrix}
    Let \(\bA\in\R^{n\times n}\) be psd and \(\bB = \bA + \bN\).
    Define
    \[
    \phimax
    \coloneqq
    \max\left\{1,\, \max_{i\in[n]} \frac{\bA_{ii}}{|\bB_{ii}|}\right\}.
    \]
    For each \(i,j\in[n]\), define the \emph{clipped matrix}
\[
    \widetilde{\bB}_{ij}
    =
    \operatorname{clip}\left(
        \bB_{ij};\, \phimax\sqrt{|\bB_{ii}|\,|\bB_{jj}|}
    \right),
\]
where
\[
    \operatorname{clip}(x;\,a)
    \coloneqq
    \begin{cases}
        -a & x < -a,\\
        x & -a \leq x \leq a,\\
        a & x > a.
    \end{cases}
\]
\end{definition}

The following result, embedded in the analysis of Bakshi, Chepurko, and Woodruff~\cite{bakshi_robust_2020}, shows that clipping does not increase corruption to \(\bA\).
We provide a self-contained proof for completeness.

\begin{proposition}[Clipping decreases corruption; \cite{bakshi_robust_2020}]
    \label{lem:clipping-lemma}
    Let \(\bA\in\R^{n\times n}\) be psd, and let \(\bB=\bA+\bN\).
    Let \(\widetilde{\bB}\) be the clipped matrix of \cref{def:clipped-matrix}.
    Then $\widetilde{\bB}$ can be written as
    \[
        \widetilde{\bB} = \bA+\widetilde{\bN} \quad \text{where } \fronorm{\widetilde{\bN}} \leq \fronorm{\bN}.
    \]
\end{proposition}

\begin{proof}
    We first note that
\(
    \bA_{ii}\leq \phimax|\bB_{ii}|
\) and \(
    \bA_{jj}\leq \phimax|\bB_{jj}|
\)
by the definition of \(\phimax\). Moreover, the off-diagonal inequality (\cref{fact:psd-inequality}) shows that \(\bA_{ij}^2\leq \bA_{ii}\bA_{jj}\), so
\[
    |\bA_{ij}|
    \leq
    \phimax\sqrt{|\bB_{ii}|\,|\bB_{jj}|}.
\]
In particular, \(\bA_{ij}\) lies in the clipping interval for entry \(ij\), so clipping at entry \(ij\) does not increase the distance of the entry to \(\bA_{ij}\). Thus,
\[
    (\widetilde{\bB}_{ij}-\bA_{ij})^2
    \leq
    \left(\bB_{ij}-\bA_{ij}\right)^2
    =
    \bN_{ij}^2.
\]
Summing over all entries gives \(\fronorm{\widetilde{\bB} - \bA}^2\leq \fronorm{\bN}^2\), which is the desired result.
\end{proof}

Although \(\widetilde{\bB}\) is not necessarily psd, the clipping threshold gives us a modified form of the off-diagonal inequality (\Cref{fact:psd-inequality}):
\[
    \widetilde{\bB}_{ij}^2
    \leq
    \phimax^2 |\bB_{ii}|\,|\bB_{jj}|\qquad\text{for all \(i,j\in[n]\)},
\]
which we use in \cref{alg:robust-sampler} to perform length-squared sampling on \(\widetilde{\bB}\).

\begin{algorithm}[t]
\caption{Fast Length-Squared Sampling for Corrupted Psd Matrix} \label{alg:robust-sampler}
\begin{algorithmic}[1]
\Require Clipped matrix $\widetilde{\bB}$ and bound $\phimax$.
\Ensure Index $I$ sampled from the length-squared distribution $\mathcal{L}(i) = \|\widetilde{\bB}_{i*}\|_2^2 / \fronorm{\widetilde{\bB}}^2$.
\State Define the diagonal sampling distribution $\mathcal{D}(i) \coloneqq |\widetilde{\bB}_{ii}| / D$, where $D \coloneqq \sum_i |\widetilde{\bB}_{ii}|$.
\While{\texttt{true}}
\State Sample independent $I, J \sim \mathcal{D}$.
\State \textbf{With probability} $\widetilde{\bB}_{IJ}^2 / (\phimax^2 |\widetilde{\bB}_{II}| |\widetilde{\bB}_{JJ}|)$, \textbf{return} $I$.
\EndWhile
\end{algorithmic}
\end{algorithm}

\begin{lemma}[Fast length-squared sampling with corruption]
\label{lem:robust-sampler}
Let $\bB = \bA + \bN$ be a corruption of a psd matrix, let $\phimax$ be as in \cref{eq:assumption-3}, and let $\widetilde{\bB}$ be the clipped matrix of \cref{def:clipped-matrix}. 
Then \cref{alg:robust-sampler} outputs \(I\) from the length-squared distribution of $\widetilde{\bB}$ and queries \(O(\phimax^2 n)\) entries of \(\bB\) in expectation.
\end{lemma}

\begin{proof}
The proof of correctness for the sampling procedure is identical to the argument in \cref{thm:fast-l2-sampling}, with the probability of accepting on any given round being
\begin{equation*}
    \beta = \sum_{i,j \in [n]} \frac{\Bbad_{ij}^2}{\phimax^2|\Bbad_{ii}||\Bbad_{jj}|} \cdot \frac{|\Bbad_{ii}|}{D}\cdot \frac{|\Bbad_{jj}|}{D} = \frac{\fronorm{\Bbad}^2}{\phimax^2 D^2}.
\end{equation*}
Here, $D \coloneqq \sum_{i \in [n]} |\Bbad_{ii}|$.
Thus, the expected number of rounds is at most
\begin{equation*}
    \frac{1}{\beta} = \frac{\phimax^2 D^2}{\fronorm{\Bbad}^2} \le \frac{\phimax^2 (\sum_{i \in [n]} |\Bbad_{ii}|)^2}{\sum_{i \in [n]} \Bbad_{ii}^2} \le \phimax^2 n. \qedhere
\end{equation*}
\end{proof}

By using the robust sampler (\cref{alg:robust-sampler}) in place of \cref{alg:sampling-procedure} in the additive-error low-rank approximation algorithm (\cref{alg:sublinear-psd-low-rank}), we obtain our algorithm for robust low-rank approximation. 

\begin{theorem}[Robust low-rank approximation]
\label{thm:robust-additive-lra}
Let \(\bA\in\R^{n\times n}\) be psd. Consider a corrupted matrix \(\bB=\bA+\bN\) satisfying assumptions \cref{eq:assumption-1,eq:assumption-3}.
Consider running \cref{alg:sublinear-psd-low-rank} on the clipped matrix $\Bbad$, using \cref{alg:robust-sampler} in place of \cref{alg:sampling-procedure} in line 2.
This procedure returns matrices \(\bX,\bY\in\R^{n\times k}\) such that, with probability at
least \(2/3\),
\[
    \fronorm{\bA-\bX\bY^\top}^2
    \leq
    \fronorm{\bA-\bA_k}^2
    +
    (\epsilon+9\sqrt{\eta})\fronorm{\bA}^2.
\]
The algorithm queries
$O\left(\frac{\phimax^2 n k}{\epsilon^2}\right)$
entries of \(\bB\) and uses $
    O\left(
        n\left(\frac{k}{\epsilon^2}\right)^{\omega-1}
        +
        \frac{\phimax^2 n k}{\epsilon^2}
    \right)$
arithmetic operations in expectation. 
\end{theorem}

\begin{proof}
Applying \cref{thm:sublinear-psd-low-rank} to $\widetilde{\bB}$, with probability at least 2/3, we get
\(\bX,\bY\in\R^{n\times k}\) such that
\begin{equation} \label{eq:robust-lra-1}
    \fronorm{\widetilde{\bB}-\bX\bY^\top}^2
    \leq
    \fronorm{\Bbad-\Bbad_k}^2
    +
    (\epsilon/8)\fronorm{\Bbad}^2
\end{equation}
and the stated resource estimates hold.
We have chosen the implicit constants larger to yield a prefactor of $\epsilon/8$ in \cref{eq:robust-lra-1}.

Finally, we transfer the bound of \cref{eq:robust-lra-1} from \(\widetilde{\bB}\) back to \(\bA\), using an argument embedded in the analysis of Bakshi, Chepurko, and Woodruff~\cite{bakshi_robust_2020}.
If \(\eta=0\), then \(\widetilde{\bB}=\bA\), and the desired guarantee follows
immediately.
Now assume \(\eta>0\).
We first bound
\begin{equation} \label{eq:robust-lra-2}
    \begin{split}
    \fronorm{\bA-\bX\bY^\top}^2
    &=
    \fronorm{(\widetilde{\bB}-\bX\bY^\top)-(\widetilde{\bB}-\bA)}^2 \\
    &= \fronorm{\widetilde{\bB}-\bX\bY^\top}^2 - 2 \tr[(\widetilde{\bB}-\bX\bY^\top)(\Bbad - \bA)] + \fronorm{\Bbad - \bA}^2 \\
    &\le \fronorm{\widetilde{\bB}-\bX\bY^\top}^2 + 2 \fronorm{\widetilde{\bB}-\bX\bY^\top}\fronorm{\Bbad - \bA} + \fronorm{\Bbad - \bA}^2 \\
    &\leq
    (1+\sqrt{\eta})\fronorm{\widetilde{\bB}-\bX\bY^\top}^2
    +
    \left(\frac{1}{\sqrt{\eta}}+1\right)
    \fronorm{\widetilde{\bB}-\bA}^2.
    \end{split}
\end{equation}
The first inequality is Cauchy--Schwarz for the trace inner product, and the second is the AM--GM inequality, which implies $2ab \leq \delta a^2+b^2/\delta$.
Now, we recall \(\fronorm{\widetilde{\bB}-\bA}^2\leq \fronorm{\bN}^2\leq \eta\fronorm{\bA}^2\) due to \cref{lem:clipping-lemma} and assumption \cref{eq:assumption-1}.
Thus
\[
\begin{aligned}
    \fronorm{\bA-\bX\bY^\top}^2 &\leq (1+\sqrt{\eta})
    \fronorm{\widetilde{\bB}-\bX\bY^\top}^2
    +
    (\sqrt{\eta} + \eta)\fronorm{\bA}^2\\
    &\leq
    (1+\sqrt{\eta})
    \left(
    \fronorm{\Bbad-\Bbad_k}^2
    +
    (\epsilon/8)\fronorm{\Bbad}^2
    \right)
    +
    2\sqrt{\eta}\fronorm{\bA}^2 \\
    &\leq
    (1+\sqrt{\eta})
    \left(
        \fronorm{\widetilde{\bB}-\bA_k}^2
        +
        (\epsilon/8)\fronorm{\Bbad}^2
    \right)
    +
    2\sqrt{\eta}\fronorm{\bA}^2.
\end{aligned}
\]
For the second inequality, we use \cref{eq:robust-lra-1} and the assumption that $\eta \leq 1$.
For the third inequality, we use that \(\widetilde{\bB}_k\) is the optimal rank-$k$ approximation to \(\widetilde{\bB}\).

It remains to compare \(\fronorm{\widetilde{\bB}-\bA_k}^2\) and
\(\fronorm{\bA-\bA_k}^2\).
Expanding and bounding the cross term as in \cref{eq:robust-lra-2}, we obtain
\[
\begin{aligned}
    \fronorm{\widetilde{\bB}-\bA_k}^2
    &=
    \fronorm{(\bA-\bA_k)+(\widetilde{\bB}-\bA)}^2 \\
    &\leq
    (1+\sqrt{\eta})\fronorm{\bA-\bA_k}^2
    +
    \left(1+\frac{1}{\sqrt{\eta}}\right)
    \fronorm{\widetilde{\bB}-\bA}^2 \\
    &\leq
    (1+\sqrt{\eta})\fronorm{\bA-\bA_k}^2
    +
    2\sqrt{\eta}\fronorm{\bA}^2.
\end{aligned}
\]
Since \(\fronorm{\bA-\bA_k}\leq \fronorm{\bA}\) and $\eta \le 1$, this implies
\[
    (1+\sqrt{\eta})\fronorm{\widetilde{\bB}-\bA_k}^2 \le (1+\sqrt{\eta})^2\fronorm{\bA-\bA_k}^2 + 2(1+\sqrt{\eta})\sqrt{\eta}\fronorm{\bA}^2
    \leq
    \fronorm{\bA-\bA_k}^2
    +
    7\sqrt{\eta}\fronorm{\bA}^2.
\]
Also,
\[
    \fronorm{\Bbad}^2
    \leq
    \left(\fronorm{\bA}+\fronorm{\widetilde{\bB}-\bA}\right)^2
    \leq
    (1+\sqrt{\eta})^2\fronorm{\bA}^2
    \leq
    4\fronorm{\bA}^2,
\]
where we again used that \(\eta\leq 1\). Combining our various bounds gives
\[
    \fronorm{\bA-\bX\bY^\top}^2
    \leq
    \fronorm{\bA-\bA_k}^2
    +
    (\epsilon+9\sqrt{\eta})\fronorm{\bA}^2.
\]
This completes the proof.
\end{proof}

\section{Acknowledgements}

ChatGPT 5.6 Sol and Claude Opus 5 were used interactively by the authors to help complete the proofs of \cref{thm:optimality,thm:frob_norm_optimality}, and to search for  related work and technical references.
All other proofs, the algorithms, and the writing is solely due to the authors.
CM and RB were partially supported by NSF grants AF-2427362 and AF-2046235.
ENE is supported by the Miller Institute for Basic Research in Science, University of California Berkeley.

\bibliographystyle{halpha}
{
\bibliography{references-2}
}

\begin{appendices}
\section{Additional Results and Omitted Proofs}
\subsection[Lp Sampling]{Extension to $\ell_p$ Sampling}
\label{app:lp-sampling}

For \(p\in[0,\infty)\) and a vector $\bv x \in \R^{n}$, let $\norm{\bv x}_p^p = \sum_{i \in [n]} |\bv x|^p$ be the standard vector $\ell_p$ norm. Define the \(\ell_p\) column norm distribution of a nonzero matrix \(\bA\in\R^{n\times n}\) as
\[
   \Pr[I = i] =  \frac{\|\bA_{*,i}\|_p^p}
    {\sum_{k=1}^n\|\bA_{*,k}\|_p^p}\qquad\text{for all $i\in[n]$}.
\]
When \(p=2\), this is the length-squared distribution of \eqref{eq:length-squared}. 
While less common than length-squared sampling, $\ell_p$ norm sampling is also applied in many randomized algorithms for matrix problems \cite{dasgupta2009sampling,munteanuturnstile}.
The rejection sampling procedure of \Cref{alg:sampling-procedure} extends directly to this distribution by replacing each diagonal weight \(\bA_{ii}\) with \(\bA_{ii}^{p/2}\).

\begin{algorithm}[H]
\caption{Fast $p$th-power entry and $\ell_p$ sampling} \label{alg:lp-sampling-procedure}
\begin{algorithmic}[1]
\Require Nonzero psd matrix $\bA \in \R^{n\times n}$, power $p \ge 0$.
\Ensure Index $(I,J)$ sampled from the $p$th-power entry distribution $\mathcal E_p(i,j) = |\bA_{ij}|^p / \sum_{k,\ell} |\bv A_{kl}|^p$. The indices $I$ and $J$ are sampled individually from the $\ell_p$ norm distribution $\mathcal{L}_p(i) = \norm{\bA_{*,i}}_p^p / \sum_{k=1}^n\|\bA_{*,k}\|_p^p$.
\State Read the diagonal of \(\bA\), and define the distribution
\(
    \mathcal D_p(i)={\bA_{ii}^{p/2}}/\sum_{k=1}^n \bv A_{kk}^{p/2}
\)
for all $i\in[n].$
\While{\texttt{true}}
\State Sample independent $I,J \sim \mathcal D_p$.
\State {With probability} $|\bA_{IJ}|^p / (\bA_{II}^{\vphantom{2}} \bA_{JJ}^{\vphantom{2}})^{p/2}$ \textbf{return} $I, J$.
\EndWhile
\end{algorithmic}
\end{algorithm}

As with \Cref{alg:sampling-procedure} for length-squared sampling, we show \Cref{alg:lp-sampling-procedure} runs in $O(n)$ expected time.

\begin{theorem}
\label{thm:fast-lp-sampling}
Fix \(p\ge 0\), and let \(\bA\in\R^{n\times n}\) be a nonzero psd matrix. Let $(I,J)$ denote the pair returned by \cref{alg:lp-sampling-procedure} on inputs $\bA$ and $p$. For all $i,j\in [n]$, we have
\[
    \Pr[I=i, J=j]
    =
    \frac{|\bA_{ij}|^p}
    {\sum_{k,\ell\in[n]}|\bA_{k\ell}|^p}.
\]
In turn, 
$
    \Pr[I = i]
    =
    \frac{\|\bA_{*,i}\|_p^p}
    {\sum_{k=1}^n\|\bA_{*,k}\|_p^p}.
$ and $
    \Pr[J = j]
    =
    \frac{\|\bA_{*,j}\|_p^p}
    {\sum_{k=1}^n\|\bA_{*,k}\|_p^p}.
$
The algorithm requires \(2n\) entry queries and $O(n)$ arithmetic operations in expectation.
    Moreover, with probability $1-\delta$, the algorithm terminates using at most $n + \lceil n \log(1/\delta) \rceil$ entry accesses and $O(n \log(1/\delta))$ arithmetic operations.
    \end{theorem}

\begin{proof}
The proof is analogous to that of \Cref{thm:fast-l2-sampling}. Since $\bv A$ is nonzero psd, with nonnegative diagonal entries, the distribution $\mathcal D_p$ computed in line 1 of \Cref{alg:lp-sampling-procedure} is valid. Additionally, by the psd off-diagonal inequality (\Cref{fact:psd-inequality}), $
    |\bA_{ij}|^p
    \leq
    (\bA_{ii}\bA_{jj})^{p/2},
$
so the acceptance probability on line 4 is also well-defined. 

Let $Z_p = \sum_{k=1}^n \bv A_{kk}^{p/2}$. Fix a pair \((i,j)\).
    In one round, the probability that the algorithm returns \((i,j)\) is
\[
    \frac{\bA_{ii}^{p/2}}{Z_p}
    \cdot
    \frac{\bA_{jj}^{p/2}}{Z_p}
    \cdot
    \frac{|\bA_{ij}|^p}{(\bA_{ii}\bA_{jj})^{p/2}}
    =
    \frac{|\bA_{ij}|^p}{Z_p^2}.
\]
Thus, for every pair \((i,j)\), the probability of outputting that pair in one round is \(|\bA_{ij}|^p/Z_p^2\). Summing over all pairs, the probability of acceptance in one round is
\[
    \beta_p
    =
    \frac{\sum_{k,\ell\in[n]}|\bA_{k\ell}|^p}{Z_p^2}.
\]
Conditioning on acceptance gives $
    \Pr[I=i, J=j]
    =
    \frac{|\bA_{ij}|^p}
    {\sum_{k,\ell\in[n]}|\bA_{k\ell}|^p},$
as desired.
Summing over the second coordinate and using symmetry gives the claimed expressions for $\Pr[I=i]$ and $\Pr[J=j]$. Finally,
\begin{align*}
    \beta_p
    =
    \frac{\sum_{k,\ell\in[n]}|\bA_{k,\ell}|^p}
    {\left(\sum_{k=1}^n\bA_{kk}^{p/2}\right)^2} \geq
    \frac{\sum_{k=1}^n\bA_{kk}^p}
    {\left(\sum_{k=1}^n\bA_{kk}^{p/2}\right)^2} \geq
    \frac{1}{n},
\end{align*}
where the last inequality follows from Cauchy--Schwarz. Hence, as in \Cref{thm:fast-l2-sampling}, the number of rounds executed is a geometric random variable with mean $1/\beta_p \le n$. The stated query complexity and runtime estimates follow, exactly as in the proof of \Cref{thm:fast-l2-sampling}.
\end{proof}

\subsection{Proof of \cref{lem:probabilistic-right-factor}}
\label{app:probabilistic-right-factor-proof}

\begin{proof}
Define the i.i.d.\ random matrices
\[
    \bR_r=\frac{\bX_{I_r,*}^{\top}\bA_{I_r,*}}{q_{I_r}}.
\]
A simple calculation shows that $\E[\bR_r] = \bX^\top \bA$ for $r \in [t]$.

Expanding the expectation, we have
\begin{equation*}
    \E\fronorm{\bY^\top  -\bX^\top\bA}^2
    =
    \E\fronorm*{
        \frac{1}{t}\sum_{r=1}^t
        \left(\bR_r-\bX^\top\bA\right)
    }^2 =
    \frac{1}{t^2}
    \sum_{r=1}^t\sum_{s=1}^t
    \E\operatorname{tr}\left[
        \left(\bR_r-\bX^\top\bA\right)^\top
        \left(\bR_s-\bX^\top\bA\right)
    \right].
\end{equation*}
The cross terms $r\ne s$ vanish by independence and the identity $\E[\bR_r] = \bX^\top \bA$.
Therefore,

\begin{equation} \label{eq:Z-err-1}
    \E\fronorm{\bY^\top  -\bX^\top\bA}^2
    =
    \frac{1}{t^2}
    \sum_{r=1}^t
    \E\fronorm{\bR_r-\bX^\top\bA}^2 =
    \frac{1}{t}
    \E\fronorm{\bR_1-\bX^\top\bA}^2 \leq
    \frac{1}{t}\E\fronorm{\bR_1}^2.
\end{equation}
We may bound $\E\fronorm{\bR_1}^2$ as follows
\begin{equation} \label{eq:Z-err-2}
    \E\fronorm{\bR_1}^2
    =
    \sum_{i:q_i>0}
    q_i
    \fronorm*{
        \frac{\bX_{i,*}^{\top}\bA_{i,*}}{q_i}
    }^2  =
    \sum_{i:q_i>0}
    \frac{\fronorm{\bX_{i,*}^{\top}\bA_{i,*}}^2}{\norm{\bX_{i,*}}^2/k} =k\sum_{i:q_i>0}\norm{\bA_{i,*}}^2\leq k\fronorm{\bA}^2 .
\end{equation}
Combining \cref{eq:Z-err-1}, \cref{eq:Z-err-2}, and Markov's inequality, we conclude
\begin{equation} \label{eq:Z-err-3}
    \fronorm{\bY^\top  -\bX^\top\bA}^2
    \leq
    \frac{k}{\delta t}\fronorm{\bA}^2\qquad\text{with probability at least \(1-\delta\).}
\end{equation}
Finally, make the orthogonal decomposition
\[
    \bA-\bX\bY^\top
    =
    (\bI-\bX\bX^\top)\bA+\bX(\bX^\top\bA-\bY^\top).
\]
By orthogonality, the squared Frobenius error decomposes
\[
    \fronorm{\bA-\bX\bY^\top}^2
    =
    \fronorm{\bA-\bX\bX^\top\bA}^2
    +
    \fronorm{\bX(\bY^\top  -\bX^\top\bA)}^2 = \fronorm{\bA-\bX\bX^\top\bA}^2
    +
    \fronorm{\bY^\top  -\bX^\top\bA}^2.
\]
In the final identity, we use that $\bX$ has orthonormal columns.
Instantiating \cref{eq:Z-err-3} completes the proof. 
\end{proof}

\subsection{Proof of \cref{thm:sublinear-psd-eigenvalue}}
\label{sec:eigenvalue-proof}

\begin{algorithm}[ht]
\caption{Sublinear PSD Eigenvalue and Eigenvector Approximation} \label{alg:sublinear-psd-eigenvalue}
\begin{algorithmic}[1]
\Require Nonzero psd matrix $\bA \in \R^{n\times n}$, accuracy parameter $\epsilon \in (0,1)$
\Ensure Estimates $\tilde\lambda_1 \geq \tilde\lambda_2 \geq \cdots \geq \tilde\lambda_n \geq 0$ for the eigenvalues of $\bA$
\State Let $C>0$ be a sufficiently large absolute constant, and set
\(
    m = \left\lceil \frac{C}{\epsilon^2} \log^4 n \log^2 \frac{1}{\epsilon} \right\rceil
\)
\State Run \cref{alg:frobenius-norm-approximation} with accuracy parameter $\epsilon/4$, producing an estimate $\widehat{F}$ for $\fronorm{\bA}^2$
\State\label{step0} Sample $K \sim \Poisson(m)$
\State \label{step1} Run \cref{alg:sampling-procedure} $K$ times on $\bA$, producing the ordered list of indices $S=\{i_1,\ldots,i_K\}$ ($i_p \in [n]$ for all $p \in [K]$) sampled by squared row norms.
\State Query the full rows $\bA_{i_1,*},\ldots,\bA_{i_K,*}$ and, for each $q\in[K]$, set
\(
    \tilde{p}_{i_q} \coloneqq {\|\bA_{i_q,*}\|_2^2}/{\widehat{F}}
\)

\State \label{step:zero1} Form the matrix $\bv A'_{S,S} \in \mathbb{R}^{K \times K}$ such that for all $(p,q) \in [K] \times [K]$ and for a sufficiently large absolute constant $c>0$:
\[
\left(\bv A'_{S,S} \right)_{pq} = \begin{cases}
0 & \text{if } p = q,\\[4pt]
0 & \text{if } p \neq q \text{ and }
\|\bv A_{i_{p},*}\|_2^2 \|\bv A_{i_{q},*}\|_2^2
\leq \dfrac{\epsilon^2 \tilde{F} |\bv A_{i_{p}i_{q}}|^2}{c\log^4 n},\\[4pt]
\frac{\bv A_{i_{p}i_{q}}}{K\sqrt{\tilde{p}_{i_p} \tilde{p}_{i_q} }} & \text{otherwise, }
\end{cases}.
\]
\State \Return eigenvalues $\widehat{\lambda}_1\geq \widehat{\lambda}_2\geq\cdots\geq \widehat{\lambda}_K$ of $\bv A'_{S,S}$ together with $n-K$ additional zeros.
\end{algorithmic}
\end{algorithm}

We first introduce the following truncation operator.
This is a simple modification of the truncation procedure in \cite[Def.~7.1]{swartworth_tight_2024}, except that we introduce a new parameter $F$ in order to accommodate an \emph{estimate} for \(\fronorm{\bA}^2\). 

\begin{definition}[Truncation operator]
\label{def:eigenvalue-zeroed-matrix}
Let \(\bA\in\R^{n\times n}\) be symmetric, let
\(\epsilon\in(0,1)\), and let \(F>0\) be an estimated value for $\fronorm{\bA}^2$. The truncated matrix
\(
    \mathsf Z_{\epsilon,F}(\bA)\in\R^{n\times n}
\)
is defined entrywise as
\[
    \bigl(\mathsf Z_{\epsilon,F}(\bA)\bigr)_{ij}
    \coloneqq
    {
    \begin{cases}
    0,
        & i=j \text{ and } \|\bA_{i, *}\|^2_2\leq \frac{\epsilon^2}{4} F,\\
    0,
        & i\neq j
        \text{ and }
        \|\bA_{i,*}\|_2^2\|\bA_{j,*}\|_2^2
        \leq
        \displaystyle
        \frac{\epsilon^2 F|\bA_{ij}|^2}
        {c\log^4 n},\\
    \bA_{ij},
        & \text{otherwise},
    \end{cases}
    }
\]
where \(c>0\) is a fixed absolute constant.
When
\(F=\fronorm{\bA}^2\), we abbreviate
\(
    \mathsf Z_{\epsilon}(\bA)
    \coloneqq
    \mathsf Z_{\epsilon,\fronorm{\bA}^2}(\bA).
\)
\end{definition}

Following~\cite{swartworth_tight_2024}, we assume that
$p_i:=\frac{m\|\bA_{i, *}\|^2_2}{\|\bA\|_F^2}\leq 1$ for all $i \in [n]$, so
that the $p_i$ are valid sampling probabilities. Since
\cref{alg:sampling-procedure} draws rows \emph{exactly} from the
squared row-norm distribution, these are the inclusion probabilities of our
sampler even though the algorithm never learns $\|\bA\|_F^2$; the approximation
$F$ enters only in the rescaling of the sampled rows, which uses
$\tilde p_i := \frac{m\|\bA_{i, *}\|^2_2}{F}$ in place of $p_i$, and in the
truncation operator $\mathsf Z_{\epsilon,F}$. We now show that this still
recovers the spectrum of $\bA$ up to additive $\epsilon\|\bA\|_F$ error,
provided $F$ is a $(1\pm\frac\epsilon4)$ \emph{relative} approximation to
$\|\bA\|_F^2$.

\begin{lemma}
\label{lem:sw-row-norm-eigenvalue}
Let \(\bA\in\R^{n\times n}\) be a nonzero symmetric matrix, let
\({\epsilon\in(0,1)}\), and
let $F\in\R$ be
such that $(1-\frac{\epsilon}{4})\|\bA\|_F^2 \leq F \leq (1+\frac{\epsilon}{4})\|\bA\|_F^2$.
Let $m \geq \frac{c}{\epsilon^2}\log^4 n\log^2\frac{1}{\epsilon}$ for some
sufficiently large constant $c>0$, and suppose that
$p_i:=\frac{m\|\bA_{i, *}\|^2_2}{\|\bA\|_F^2} \leq 1$ for all $i \in [n]$.
Let $\bS \in \R^{k \times n}$ be the scaled sampling matrix which samples each
row $i \in [n]$ of $\bA$ independently with probability $p_i$ and
rescales the
sampled row by $\frac{1}{\sqrt{\tilde p_i}}$
where  $\tilde p_i=\frac{m\|\bA_{i, *}\|^2_2}{F}$. Let
$\mathsf Z_{\epsilon,F}(\bA)$ be the zeroed out matrix as in
Definition~\ref{def:eigenvalue-zeroed-matrix}.
Then, with probability at least $\frac{5}{6}$, the eigenvalues of
\(
    \bS\mathsf Z_{\epsilon,F}(\bA)\bS^\top
\)
together with \(n-k\) additional zeros form an additive
\(\epsilon\|\bA\|_F\)-approximation to the
spectrum of \(\bA\).
\end{lemma}
\begin{proof}
Observe that when $F=\|\bA\|_F^2$, the Lemma is exactly
Lemma 7.7 of~\cite{swartworth_tight_2024}. We now show that when $F$ is only a
$(1\pm\frac\epsilon4)$ relative approximation to $\|\bA\|_F^2$, the procedure in
the statement is, after an explicit rescaling of the accuracy parameter and of
the sampling matrix, precisely the procedure analyzed there, so that the same
argument applies. We will prove the lemma for
$\epsilon\leq\frac45$, which is sufficient for our application. Note that 
any algorithm may always replace $\epsilon$
by $\min\{\epsilon,\frac45\}$, which only strengthens the guarantee and changes
$m$ by a constant factor.
 Let
\[
    \theta
    :=
    \frac{\sqrt F}{\|\bA\|_F},
    \qquad\text{so that}\qquad
    \theta^2
    =
    \frac{F}{\|\bA\|_F^2}
    \in\left[1-\frac\epsilon4,1+\frac\epsilon4\right]
    \subseteq\left[\tfrac45,\tfrac65\right],
\]
where we used $\epsilon\leq\frac45$. Let $\bS_{\mathrm{ex}}\in\R^{k\times n}$
denote the sampling matrix which includes each row $i$ independently with
probability $p_i$ and rescales it by $1/\sqrt{p_i}$, i.e.\ the sampling matrix
used in Lemma 7.7 of~\cite{swartworth_tight_2024}, and set
\[
    \epsilon'
    :=
    \epsilon\theta
    =
    \frac{\epsilon\sqrt F}{\|\bA\|_F}.
\]
Since $\epsilon\leq\frac45$ and $\theta\leq\sqrt{6/5}$, we have
$0<\epsilon'\leq\frac45\sqrt{6/5}<1$, so $\epsilon'$ is an admissible accuracy
parameter. 
We now compare the two procedures and then account for the
resulting distortion. 

\medskip
\noindent\emph{The truncated matrices coincide,
$\mathsf Z_{\epsilon,F}(\bA)=\mathsf Z_{\epsilon'}(\bA)$.}
In Definition~\ref{def:eigenvalue-zeroed-matrix}, the accuracy parameter and $F$
enter both zeroing conditions only through the product (accuracy)$^2\cdot F$,
and by construction $
    \epsilon^2 F
    =
    \epsilon^2\theta^2\|\bA\|_F^2
    =
    (\epsilon')^2\|\bA\|_F^2$. Thus, the two truncation operators $\mathsf Z_{\epsilon,F}$ 
and $\mathsf Z_{\epsilon',\|\bA\|_F^2}$ are identical, 
and the latter is precisely the operator $\mathsf Z_{\epsilon'}$ used 
in Lemma 7.7 of~\cite{swartworth_tight_2024}.

\medskip
\noindent\emph{The sampling matrices differ by a scalar.}
The two rescalings differ by the ratio
$\sqrt{p_i/\tilde p_i}=\sqrt F/\|\bA\|_F=\theta$, which is the \emph{same} for
every $i\in[n]$ since the column norms cancel. Hence $\bS=\theta\,\bS_{\mathrm{ex}}$,
and combining with Step 1,
\begin{equation}
\label{eq:sw-scaled-submatrix}
    \bS\mathsf Z_{\epsilon,F}(\bA)\bS^\top
    =
    \theta^2\,
    \bS_{\mathrm{ex}}\mathsf Z_{\epsilon'}(\bA)\bS_{\mathrm{ex}}^\top .
\end{equation}

\medskip
Note that since $\epsilon$ and $\epsilon'$ differ by at most a constant 
factor, the sample complexity $m=\Omega(\frac{\log^4 n}{\epsilon'^2}\log^2 \frac{1}{\epsilon'})
=\Omega(\frac{\log^4 n}{\epsilon^2}\log^2 \frac{1}{\epsilon})$.
By the two steps above, the hypotheses of Lemma 7.7 of~\cite{swartworth_tight_2024} hold
for $\bA$ with accuracy parameter $\epsilon'\in(0,1)$ and sample complexity
$m$. Let $\nu_1\geq\cdots\geq\nu_n$ denote the eigenvalues of
$\bS_{\mathrm{ex}}\mathsf Z_{\epsilon'}(\bA)\bS_{\mathrm{ex}}^\top$ together with
$n-k$ additional zeros. Then, with probability at least $\frac56$,
\begin{equation}
\label{eq:sw-exact-guarantee}
    |\nu_i-\lambda_i(\bA)|
    \leq
    \epsilon'\|\bA\|_F
    =
    \epsilon\theta\|\bA\|_F
    \qquad\text{for all }i\in[n].
\end{equation}
By~\eqref{eq:sw-scaled-submatrix} and $\theta^2>0$, the eigenvalues of
$\bS\mathsf Z_{\epsilon,F}(\bA)\bS^\top$ together with $n-k$ additional zeros
are exactly $\tilde\lambda_i=\theta^2\nu_i$, in the same order. Since
$\sum_{i}\lambda_i(\bA)^2=\|\bA\|_F^2$, we have $|\lambda_i(\bA)|\leq\|\bA\|_F$
for every $i$, so on the event~\eqref{eq:sw-exact-guarantee} the triangle
inequality gives, for every $i\in[n]$,
\[
    \left|\tilde\lambda_i-\lambda_i(\bA)\right|
    \leq
    \theta^2\left|\nu_i-\lambda_i(\bA)\right|
    +
    \left|\theta^2-1\right|\left|\lambda_i(\bA)\right|
    \leq
    \theta^3\epsilon\|\bA\|_F
    +
    \left|\theta^2-1\right|\|\bA\|_F .
\]
The first term is the sampling error, inflated by the rescaling; the second is
the systematic multiplicative distortion caused by using $F$ in place of
$\|\bA\|_F^2$. For the first term, $\theta^2\leq\frac65$ and
$\left(\frac65\right)^3=\frac{216}{125}<\frac{16}{9}=\left(\frac43\right)^2$
give $\theta^3\leq\frac43$; for the second,
$\theta^2\in[1-\frac\epsilon4,1+\frac\epsilon4]$ gives
$|\theta^2-1|\leq\frac\epsilon4$. Hence, for every $i\in[n]$,
\[
    \left|\tilde\lambda_i-\lambda_i(\bA)\right|
    \leq
    \left(\frac43+\frac14\right)\epsilon\|\bA\|_F
    =
    \frac{19}{12}\epsilon\|\bA\|_F
    \leq
    \frac53\epsilon\|\bA\|_F .
\]
Adjusting $\epsilon$ by a constant factor, we obtain the final bound
as stated.
\end{proof}

Our main result 
Theorem~\ref{thm:sublinear-psd-eigenvalue} now follows from the same \emph{Poissonization} argument as in 
the proof of Theorem 5 in~\cite{RB26}.
We briefly recall the 
high level ideas here and refer the reader to~\cite{RB26} for the details. 
There are two main steps in the proof.
First, following~\cite{swartworth_tight_2024} and~\cite{RB26}, 
we consider an 
\emph{inflated} matrix $\bv A_M$, where each column and row 
of $\bv A$ is duplicated $M$ times and 
scaled by $1/\sqrt{M}$ to ensure all 
sampling probabilities are at most $1$.
Note that the nonzero spectrum of $\bv A_M$ is the same as that of $\bv A$ and so, $\|\bv A_M \|_F=\|\bv A \|_F$.
Thus, we can apply
\cref{lem:sw-row-norm-eigenvalue} to $\bv A_M$
and obtain additive approximations to the eigenvalues of $\bv A$.
Also, observe that sampling from the inflated matrix $\bv A_M$ is
equivalent to sampling each row $i \in [n]$ from $\bv A$
$\Binomial(M, \frac{m\pi_i}{M})$ times where
$\pi_i=\frac{\|\bv A_{i, *}\|_2^2}{\|\bv A\|_F^2}$. Thus, the error guarantee holds using this Binomial sampling scheme. Note that, since $\bv A_M$ has $nM$ rows, using \cref{lem:sw-row-norm-eigenvalue}, the numbers of rows sampled $m$ should be $ m = \frac{c\log^4 nM}{\epsilon^2}\log^2 \frac{1}{\epsilon}$. As we explain below, we will set $M=O(m^2)$, so, $m=O \left(\frac{\log^4 n}{\epsilon^2}\log^2 \frac{1}{\epsilon} \right)$ after adjusting the constants.

Second, we transfer this error guarantee to actual
the Poisson sampling scheme used in 
\cref{alg:sublinear-psd-eigenvalue} by 
bounding the total variation distance between 
the Binomial and Poisson sampling distributions.
To do this, we use Lemma 17 of~\cite{RB26} which bounds the total variation distance between the two distributions. Specifically, we can show that the sampling scheme in steps 3 and 4 of~\cref{alg:sublinear-psd-eigenvalue} is the same as sampling row $i$ $\mathrm{Poisson}(m\pi_i)$
times independently of other rows.
Thus, if $n_i \sim \mathrm{Poisson}(m\pi_i)$, 
$(n_1, \ldots, n_n)$ are \emph{mutually independent} (Lemma 18 of~\cite{RB26}).
By Lemma 17 of~\cite{RB26},
the total variation distance between the two joint distributions of $\mathrm{Poisson}(m\pi_i)$ ($\forall i \in [n]$)
and $\mathrm{Binomial}(M, m\pi_i/M)$ ($\forall i \in [n]$)
is at most $m^2/M$,
which is bounded by some small constant $\delta \in (0,1)$ for $M \geq \frac{m^2}{\delta}$. We refer the reader to~\cite{RB26} for the details.

\paragraph{Proof of \cref{thm:sublinear-psd-eigenvalue}}
We can prove the error bound using the two step reduction to the sampling scheme of \cref{alg:sublinear-psd-eigenvalue} from the sampling scheme of \cref{lem:sw-row-norm-eigenvalue} as described above. We bound the query and arithmetic complexity below.

\medskip
\noindent\textbf{Query and arithmetic complexity.}
Here $K\sim\Poisson(m)$, so $\E[K]=m$, $\E[K^2]=O(m^2)$ and $\E[K^3]=O(m^3)$.
Step 4 of \cref{alg:sublinear-psd-eigenvalue} costs $O(n/\epsilon^2)$ queries and arithmetic operations by
Theorem~\ref{thm:frobenius-norm-approximation-algorithm}; each of the $K$ calls
to \cref{alg:sampling-procedure} costs $O(n)$ queries and arithmetic
operations in expectation by Theorem~\ref{thm:fast-l2-sampling}, so their total expected cost is $O(nm)$; reading the sampled rows in step~5 costs $Kn$ queries and $O(Kn)$
arithmetic operations; forming $\bv A'_{S,S}$ costs $O(K^2)$ and its
eigendecomposition $O(K^{\omega})$ arithmetic operations. In expectation this is
\[
    O\!\left(\frac n{\epsilon^2}+nm\right)
    =O\!\left(\frac n{\epsilon^2}\log^4 n\log^2\frac1\epsilon\right)
\]
queries and
\[
    O\!\left(nm+m^{\omega}\right)
    =O\!\left(\frac n{\epsilon^2}\log^4 n\log^2\frac1\epsilon
    +\frac1{\epsilon^{2\omega}}\log^{4\omega}n\log^{2\omega}\frac1\epsilon\right)
\]
arithmetic operations.
\qed

\end{appendices}

\end{document}